\documentclass[journal]{IEEEtran}
\usepackage{cite}
\usepackage{amsmath,amssymb,amsfonts}
\usepackage{algorithm}
\usepackage{algorithmic}
\usepackage{graphicx}
\usepackage{url}
\usepackage{bm}
\usepackage{cleveref}

\usepackage{tikz}
\usetikzlibrary{arrows.meta,positioning,calc,fit,decorations.pathreplacing}
\newtheorem{proposition}{Proposition}
\newtheorem{remark}{Remark}

\begin{document}

\title{Programmable Unitary RF Networks for Narrowband Hybrid Beamforming}

\author{Nikola Zlatanov and Damir Salakhov
\thanks{N.~Zlatanov and D.~Salakhov are with Innopolis University, Innopolis, 420500, Russia (e-mails: \texttt{n.zlatanov@innopolis.ru} and \texttt{d.salakhov@innopolis.ru}).}
}

\maketitle

\begin{abstract}
Conventional hybrid beamforming architectures are compared with one another and with the fully-digital architecture under the same radiated antenna power. However, the physically relevant power budget is the power supplied by the RF chains to the passive analog beamforming network, which then transfers it to the antenna ports. Some architectures, such as fully-connected architectures, can divert part of this power to terminated ports even under ideal passive-component assumptions, thereby delivering less power to the antenna ports than that supplied by the RF chains. Other architectures, such as partially connected architectures, transfer all RF-chain power to the antenna ports but have restricted beamforming flexibility and therefore do not attain the fully-digital sum-rate optimum in general.

Motivated by this RF-chain injected-power and beamforming flexibility viewpoint, this paper proposes a hybrid beamforming architecture that, under ideal passive-component assumptions: i) transfers all RF-chain injected power to the antenna ports without loss; and ii) with continuous phase control and sufficiently many programmable phase layers, attains the fully-digital sum-rate optimum under the same RF-chain injected-power constraint.

The proposed hybrid beamforming architecture uses a programmable unitary analog beamforming network realized by interlacing fixed mixing layers with programmable diagonal phase-shifting layers. We compare this architecture with fully-connected, partially connected, and Butler/DFT beam-selection baselines under ideal passive-component assumptions and show numerically that none of these baselines attains the fully-digital sum-rate optimum under the same RF-chain injected-power constraint. Finally, we compare the different architectures when component losses are taken into account and illustrate the balance between beamforming flexibility and accumulated attenuation.
\end{abstract}


\section{Introduction}
\label{sec:intro}
\IEEEPARstart{F}{ully-digital} multi-antenna transmitters require one radio-frequency (RF) chain per antenna, which is often infeasible in large arrays due to power, cost, and routing constraints \cite{heath2016overview}. Hybrid analog--digital beamforming architectures address this bottleneck by combining a high-dimensional analog beamforming network with a low-dimensional digital stage. Typically, these architectures comprise $r$ RF chains and $N$ antennas, with $r\ll N$, and are able to support $S\le r$ data streams.

Most prior works compare hybrid beamforming architectures at equal radiated power from their antennas. However, such a comparison does not reveal how much of the physically relevant power, which is the power injected by the RF chains, actually reaches the antenna ports and influences the data rate. Once the analog stage is interpreted as a physical analog beamforming network between the RF-chain outputs and the antenna ports, it becomes clear that not all power injected by the RF chains into the analog beamforming network reaches the antenna ports, even under ideal passive-component assumptions. As a result, comparing hybrid beamforming architectures under equal radiated power may allow different levels of power supplied by the RF chains, thereby concealing the effect of power expenditure on the data rate. On the other hand, comparing hybrid beamforming architectures under the same power supplied by the RF chains makes the comparison fair and exposes the effect of power expenditure on the data rate. 

This issue has already been highlighted in realistic RF analyses of hybrid precoders \cite{heath2016overview,garcia2016rf}, which showed that fully-connected analog beamforming networks, such as the ones in \cite{sohrabi2016hybrid,yu2019dps}, can be strongly contractive even under ideal passive-component assumptions. Consequently, only a fraction of the RF-chain power reaches the antennas, while the remaining power is absorbed in the combining/termination circuitry and is therefore unavailable for radiation. This is precisely the effect that equal-radiated-power normalization obscures. As a result, analog beamforming architectures that preserve the power supplied by the RF chains have been proposed, such as the partially connected architecture \cite{gao2016ee}, Butler/DFT analog beamforming networks \cite{garcia2016rf} and the lens- and beamspace-based front-ends \cite{heath2016overview}, at the cost of more restricted analog beamforming flexibility.

From the viewpoint of RF transport efficiency, maximal power transfer occurs when the underlying $N\times N$ forward transmission matrix is unitary. In that case, the induced $N\times r$ analog beamforming matrix is semi-unitary and transfers all RF-chain power to the antenna ports without loss. However, power preservation is only one of the two important aspects. The second important aspect is sufficient beamforming flexibility to realize an analog subspace that enables the hybrid beamforming architecture to attain the fully-digital sum-rate optimum.

This paper proposes and studies a hybrid beamforming architecture whose analog beamformer conducts power-preserving RF transport from the RF chains to the antenna ports and is programmed via its phase-shifter settings to realize the desired beamforming subspace. We show both theoretically and numerically that the proposed hybrid beamforming architecture, with sufficiently many programmable phase layers, attains the fully-digital sum-rate optimum under the same RF-chain injected-power constraint and ideal passive-component assumptions. In contrast, our numerical comparisons show that none of the fully-connected, partially connected, or Butler/DFT beam-selection baselines attains this optimum under the same conditions.

The analog beamforming network in the proposed architecture consists of $M$ programmable diagonal phase-shifting layers interlaced with $M+1$ fixed mixer layers. For an $N$-antenna transmitter with $r$ RF chains, the first $r$ inputs of this analog beamforming network are driven by the RF chains, the remaining $N-r$ inputs are terminated in matched loads, and the $N$ outputs feed the antenna ports. The analog beamforming subspace is controlled by the phase settings in the $M$ programmable layers. 

Interlaced mixer--phase architectures, consisting of fixed mixing layers interleaved with programmable diagonal phase layers, have mathematical roots in the unitary-group results of Borevich and Krupetskii~\cite{borevich1981subgroups}. An early optical realization was introduced by Morizur et al.~\cite{morizur2010programmable} for programmable spatial-mode conversion, using phase modulation interleaved with fixed Fourier transforms. This principle was subsequently used for mode multiplexing in optical-fiber communications through multi-plane light conversion~\cite{labroille2014efficient}, and later implemented at microwave frequencies for programmable analog matrix processing~\cite{keshavarz2025programmable}. Building on these developments, this paper proposes and studies a hybrid beamforming architecture whose analog beamforming network consists of fixed unitary RF mixing layers interleaved with programmable phase-shifting layers. To the best of our knowledge, this specific RF architecture has not previously been developed for hybrid digital--analog beamforming.

Very recently, MiLAC-based beamforming has been proposed using fully connected tunable-admittance analog beamforming networks, and a hybrid digital--MiLAC factorization was shown, at the matrix-factorization level, to achieve digital-beamforming flexibility with as many RF chains as data streams \cite{wu2026milac}. Our work is close in spirit because it also aims to achieve digital-beamforming flexibility, but it proposes a different architecture: instead of a dense tunable-admittance analog beamforming network, we propose and study a structured unitary cascade built from fixed mixers and programmable diagonal phase-shifting layers.

Our contributions are as follows:
\begin{enumerate}
\item We propose a novel hybrid beamforming architecture in which the analog beamforming network is built with $M$ programmable diagonal phase-shifting layers interlaced with $M+1$ fixed mixer layers. For an $N$-antenna transmitter with $r$ RF chains, the first $r$ inputs of this analog beamforming network are driven by the RF chains, the remaining $N-r$ inputs are terminated in matched loads, and the $N$ outputs feed the antenna ports. 

\item We show that the induced analog beamformer from the $r$ driven inputs to the $N$ antenna ports is semi-unitary and therefore transfers the RF-chain injected power to the antenna ports without loss, under ideal passive-component assumptions.

\item We introduce RF-chain injected power as a comparison metric and show that it is the physically relevant power metric under which different hybrid beamforming architectures should be compared.

\item For the proposed hybrid beamforming architecture, we design the digital and analog beamformers and show both theoretically and numerically that, given sufficiently many programmable phase layers $M$, it attains the fully-digital sum-rate optimum under the same RF-chain injected-power constraint and ideal passive-component assumptions. In contrast, our numerical comparisons show that none of the fully-connected, partially connected, or Butler/DFT beam-selection baselines attains this optimum under the same conditions.

\item For a smaller number of phase layers than the sufficient count, we provide a theoretical upper bound on the gap between the best attainable sum rate of the proposed hybrid beamforming architecture and the fully-digital sum-rate optimum.

\item We compare the different architectures when component losses are taken into account and illustrate the balance between beamforming flexibility and accumulated attenuation.
\end{enumerate}

\textit{Paper Scope:} 
This paper is a study of hybrid beamforming with programmable unitary analog beamforming networks under ideal passive-component assumptions in order to establish and study the architecture's maximum capabilities. We also conduct an initial numerical study with component losses and leave a comprehensive study accounting for such losses, possibly using a modified architecture, for future work. Moreover, this paper focuses on the \emph{narrowband} case and leaves the wideband case for future work.

This paper is organized as follows. Section~\ref{sec:system} introduces the narrowband system model, the power injected by the RF chains, and the proposed hybrid beamforming architecture. Section~\ref{sec:beamformer_design} presents the digital and analog beamformer design, the programming algorithm, and its computational complexity. Section~\ref{sec:depth} establishes fully-digital sum-rate attainment conditions, analyzes the required number of phase layers, and bounds the finite-layer rate gap. Section~\ref{sec:quantized} presents phase-quantized programming. Section~\ref{sec:baselines} defines the fully-digital sum-rate optimization problem and hybrid comparator architectures under equal RF-chain injected power. Section~\ref{sec:numerics} provides numerical results, and Section~\ref{sec:conclusion} concludes the paper.

\section{Proposed Hybrid Beamforming Architecture}
\label{sec:system}

\subsection{Hybrid beamforming and RF-chain injected power}
We consider a narrowband downlink transmitter with $N$ antennas and $r$ RF chains, transmitting $S$ data streams, where $S\le r\ll N$. The transmitter employs a hybrid precoder comprising a digital beamformer followed by an analog beamformer, given by
\begin{equation}
\label{eq:hyb}
\mathbf{F}=\mathbf{F}_{\mathrm{RF}}\,\mathbf{F}_{\mathrm{BB}},
\end{equation}
where $\mathbf{F}_{\mathrm{RF}}\in\mathbb{C}^{N\times r}$ is the analog beamformer matrix and $\mathbf{F}_{\mathrm{BB}}\in\mathbb{C}^{r\times S}$ is the digital beamformer matrix. 

Let $\mathbf{s}\in\mathbb C^S$ denote the information vector, with $\mathbb E[\mathbf s\mathbf s^H]=\mathbf I_S$. Digital beamforming produces from $\mathbf{s}$ the RF-chain output vector 
\begin{equation}
\label{eq_tr_x}
\mathbf{x}=\mathbf{F}_{\mathrm{BB}}\mathbf{s}\in\mathbb{C}^{r}.
\end{equation}
The analog beamformer then maps $\mathbf x$ to the antenna-excitation vector
\begin{equation}
\label{eq_tr_x_ant}
\mathbf{x}_{\mathrm{ant}}=\mathbf{F}_{\mathrm{RF}}\mathbf{x}
=\mathbf{F}_{\mathrm{RF}}\mathbf{F}_{\mathrm{BB}}\mathbf{s}.
\end{equation}

The total \emph{power injected by the RF chains} is given by
\begin{equation}
\label{eq:P_inj}
P_{\mathrm{inj}}
=
\mathbb{E}\|\mathbf{x}\|_2^2
=
\|\mathbf{F}_{\mathrm{BB}}\|_F^2.
\end{equation}
The RF chains supply this power $P_{\mathrm{inj}}$ to the analog beamforming network, which is passive and only transfers this power to the antenna ports.
For matched antennas with unit radiation efficiency, the total radiated power equals the antenna-port power,
\begin{equation}
\label{eq:P_rad}
P_{\mathrm{rad}}
\triangleq
\mathbb{E}\|\mathbf{x}_{\mathrm{ant}}\|_2^2
=
\|\mathbf{F}_{\mathrm{RF}}\mathbf{F}_{\mathrm{BB}}\|_F^2.
\end{equation}
For a general passive analog beamforming network one has $P_{\mathrm{rad}}\le P_{\mathrm{inj}}$.

\subsection{Channels and Sum-Rate}

We consider a narrowband multiuser downlink transmission to $S$ single-antenna users. Let $\mathbf{H}$ be the channel matrix, given by
$\mathbf{H}=[\mathbf{h}_1,\ldots,\mathbf{h}_S]\in\mathbb{C}^{N\times S}$,
where $\mathbf{h}_s$ is the channel to user $s$. Let $y_s$ be the received signal at user $s$, and let $\mathbf{y}$ denote the stacked received-signal vector. Then the received signal vector $\mathbf{y}$ is given by
\begin{equation}
\mathbf{y}
=
\mathbf{H}^H \mathbf{F}_{\mathrm{RF}}\mathbf{F}_{\mathrm{BB}}\mathbf{s}
+\mathbf{n},
\end{equation}
with $\mathbf{n}\sim\mathcal{CN}(\mathbf{0},\sigma^2\mathbf{I}_S)$.

The channel matrix $\mathbf H$ is known at the transmitter. We assume $\mathbf s\sim\mathcal{CN}(\mathbf0,\mathbf I_S)$, independently of the receiver noise.
Treating residual multiuser interference as noise, the SINR of user $s$ is
\begin{equation}
\label{eq:sim_sinr}
\mathrm{SINR}_s
=
\frac{\big|\mathbf{h}_s^H\mathbf{f}_s\big|^2}
{\sum_{j\neq s}\big|\mathbf{h}_s^H\mathbf{f}_j\big|^2+\sigma^2},
\end{equation}
where $\mathbf f_s$ is the $s$-th column of $\mathbf F$ in \eqref{eq:hyb}. The achievable sum rate, in bit/s/Hz, is
\begin{equation}
\label{eq:sumrate}
R_{\Sigma}
=
\sum_{s=1}^{S}\log_2\!\big(1+\mathrm{SINR}_s\big).
\end{equation}

\subsection{Lossless analog beamforming network and RF-chain inputs}
\label{subsec:hardware_clarify}
We now propose a general programmable analog beamforming network.

The programmable analog beamforming network has $2N$ physical ports: $N$ input ports on the RF-chain side and $N$ output ports on the antenna side. Its first $r$ input ports are driven by the RF chains, and the remaining $N-r$ ports are terminated in matched loads with zero incident waves. Its $N$ output ports are connected to the $N$ antennas. A diagram, where the analog beamforming network is shown as a box, is illustrated in \Cref{fig:system_diagram}; in the following subsection we show how to realize this box.

 The transmission through the analog beamforming network, from its $N$ input ports to its $N$ output ports, is represented by the forward transmission matrix $\mathbf{X}(\Phi)\in\mathbb{C}^{N\times N}$, where $\Phi$ collects programmable phase settings, which are the only variables in the analog beamforming network. We assume ideal passive components with known responses, lossless and reflection-free transmission, and matched antennas. Under these assumptions, the forward transmission matrix is unitary and satisfies
\begin{equation}
\mathbf{X}(\Phi)^H\mathbf{X}(\Phi)=\mathbf{I}_N.
\label{eq:X_matrix}
\end{equation}

Let $\mathbf u_{\mathrm{in}}\in\mathbb C^N$ denote the complex amplitudes at the $N$ input ports of the analog beamforming network. Since only the first $r$ ports are driven by the RF chains, and the remaining $N-r$ ports are terminated in matched loads with zero incident waves, 
$\mathbf u_{\mathrm{in}}=[\mathbf x^T,\mathbf 0_{N-r}^T]^T$, where $\mathbf x$ is the output vector from the $r$ RF chains, given by \eqref{eq_tr_x}.
The $N$ output ports provide the antenna-excitation vector $\mathbf x_{\mathrm{ant}}$, giving the input--output relation
\begin{equation}
\label{eq:io_modes}
\mathbf{x}_{\mathrm{ant}}=\mathbf{X}(\Phi)\,\mathbf{u}_{\mathrm{in}}.
\end{equation}

For a compact representation, let
$\mathbf{E}_r=[\mathbf{e}_1,\ldots,\mathbf{e}_r]\in\mathbb{R}^{N\times r}$
embed the RF-chain vector $\mathbf{x}\in\mathbb{C}^{r}$ into $\mathbb{C}^{N}$ via
\begin{equation}
\mathbf{u}_{\mathrm{in}}=\mathbf{E}_r\mathbf{x}.
\end{equation}
The antenna excitation vector is then
\begin{equation}\label{eq:x_ant_2}
\mathbf{x}_{\mathrm{ant}}=\mathbf{X}(\Phi)\mathbf{E}_r\,\mathbf{x},
\end{equation}
Now, combining \eqref{eq:x_ant_2} with \eqref{eq_tr_x_ant}, we obtain that the analog beamformer is given by
\begin{equation}
\label{eq:FRF}
\mathbf{F}_{\mathrm{RF}}(\Phi)=\mathbf{X}(\Phi)\mathbf{E}_r \in\mathbb{C}^{N\times r}.
\end{equation}
Because $\mathbf{X}(\Phi)$ is unitary, the analog beamformer is semi-unitary:
\begin{equation}
\label{eq:semi_unitary}
\mathbf{F}_{\mathrm{RF}}(\Phi)^H\mathbf{F}_{\mathrm{RF}}(\Phi)=\mathbf{I}_r.
\end{equation}

The analog beamformer is completely determined by the phase configuration $\Phi$ through \eqref{eq:FRF}. In the following, we use $\mathbf{F}_{\mathrm{RF}}$ and $\mathbf{F}_{\mathrm{RF}}(\Phi)$ interchangeably whenever the dependence on $\Phi$ is implicit.  

\begin{proposition}[Power preservation of the proposed analog stage]
\label{prop:power_preserve}
For any $\mathbf{x}\in\mathbb{C}^{r}$, the analog beamformer in \eqref{eq:FRF} preserves Euclidean norm:
\begin{equation}
\|\mathbf{F}_{\mathrm{RF}}(\Phi)\mathbf{x}\|_2^2=\|\mathbf{x}\|_2^2.
\end{equation}
Consequently,
\begin{equation}
P_{\mathrm{rad}}=P_{\mathrm{inj}}
\end{equation}
for the proposed architecture under the ideal passive-component assumptions.
\end{proposition}

\begin{IEEEproof}
Using \eqref{eq:semi_unitary},
\[
\|\mathbf{F}_{\mathrm{RF}}\mathbf{x}\|_2^2
=
\mathbf{x}^H\mathbf{F}_{\mathrm{RF}}^H\mathbf{F}_{\mathrm{RF}}\mathbf{x}
=
\mathbf{x}^H\mathbf{x}
=
\|\mathbf{x}\|_2^2.
\]
Taking expectation over $\mathbf{x}=\mathbf{F}_{\mathrm{BB}}\mathbf{s}$ gives
$P_{\mathrm{rad}}=P_{\mathrm{inj}}$.
\end{IEEEproof}

\begin{remark}[Physical ports and unused input terminations]
For an ideal reciprocal realization, the full scattering relation of the analog beamforming network is
\[
\begin{bmatrix}\mathbf b_L\\\mathbf b_R\end{bmatrix}
=\underbrace{\begin{bmatrix}
\mathbf 0&\mathbf X(\Phi)^T\\
\mathbf X(\Phi)&\mathbf 0
\end{bmatrix}}_{\mathsf S_{\mathrm{RF}}(\Phi)\in\mathbb C^{2N\times2N}}
\begin{bmatrix}\mathbf a_L\\\mathbf a_R\end{bmatrix},
\]
where $\mathbf a_L,\mathbf a_R$ and $\mathbf b_L,\mathbf b_R$ are incident and outgoing power-wave vectors on the RF-chain and antenna sides, respectively. Lossless, reflection-free transmission makes both $\mathsf S_{\mathrm{RF}}$ and its forward block $\mathbf X$ unitary. Moreover, it is equivalent and yet more convenient to work with $\mathbf X$ than with $\mathsf S_{\mathrm{RF}}$. Matched antennas give $\mathbf a_R=\mathbf0$, so $\mathbf b_R=\mathbf X\mathbf a_L$ and $\mathbf b_L=\mathbf0$. Hence the unused input loads absorb $P_{\mathrm{term}}=0$ from the transmitted signal. 
\end{remark}

\subsection{Realizing $\mathbf{X}(\Phi)$ via the interlaced mixer--phase architecture}
\label{subsec:realazing_X}
The description in Subsection~\ref{subsec:hardware_clarify} applies to any analog beamforming network satisfying \eqref{eq:X_matrix}. We now specify a realization of $\mathbf X(\Phi)$ using fixed mixing layers and programmable phase layers.

We realize $\mathbf{X}(\Phi)$ as an interlaced cascade of a fixed unitary mixing layer $\mathbf{W}\in\mathbb{C}^{N\times N}$, satisfying $\mathbf{W}^H\mathbf{W}=\mathbf I_N$, and $M$ programmable diagonal phase-shifting layers, as
\begin{equation}
\label{eq:X}
\mathbf{X}(\Phi)=\mathbf{W}\mathbf{D}_M (\bm{\phi}_M) \mathbf{W}\mathbf{D}_{M-1} (\bm{\phi}_{M-1})\cdots \mathbf{W}\mathbf{D}_1 (\bm{\phi}_1)\mathbf{W},
\end{equation}
where
\begin{equation}
\mathbf{D}_k(\bm{\phi}_k)=\mathrm{diag}\big(e^{j\phi_{1,k}},\ldots,e^{j\phi_{N,k}}\big)
\end{equation}
and $\Phi=\{\phi_{n,k}\}$ collects all phases from all phase-shifting layers, and represents the programmable phase setting. 
Each phase layer satisfies $\mathbf D_k^H\mathbf D_k=\mathbf I_N$.
The fixed layer $\mathbf W$ can be implemented as a normalized DFT, and $\mathbf D_k(\bm\phi_k)$ is realized by a programmable phase-shifting layer. 

The fixed layer $\mathbf W$ mixes the signals arriving at its $N$ input ports and outputs $N$ mixed signals via its $N$ output ports, which are then forwarded to the phase-shifting layer which phase-shifts each of these $N$ signals independently before forwarding them to the next mixing layer. 

 Because each factor in \eqref{eq:X} is unitary, i.e., $\mathbf W$ and $\mathbf D_k(\bm\phi_k)$ are unitary, the forward transmission matrix $\mathbf{X}(\Phi)$ remains unitary for every phase configuration $\Phi$, thereby satisfying \eqref{eq:X_matrix}.

\subsection{End-to-end transmitter architecture}
As explained in \Cref{subsec:hardware_clarify}, 
the analog beamformer $\mathbf{F}_{\mathrm{RF}}$ is realized by driving $r$ of the $N$ input ports of the analog beamforming network and terminating the remaining $N-r$ inputs in matched loads. Next, as explained in \Cref{subsec:realazing_X}, the analog beamforming network is realized as a cascade of unitary mixing layers interlaced with programmable diagonal phase-shifting layers, as in \eqref{eq:X}. Combining all of this with the digital beamformer and the $N$ antennas gives the proposed hybrid beamforming architecture shown in \Cref{fig:system_diagram}.
\begin{figure}[t]
  \centering
  \includegraphics[width=\columnwidth]{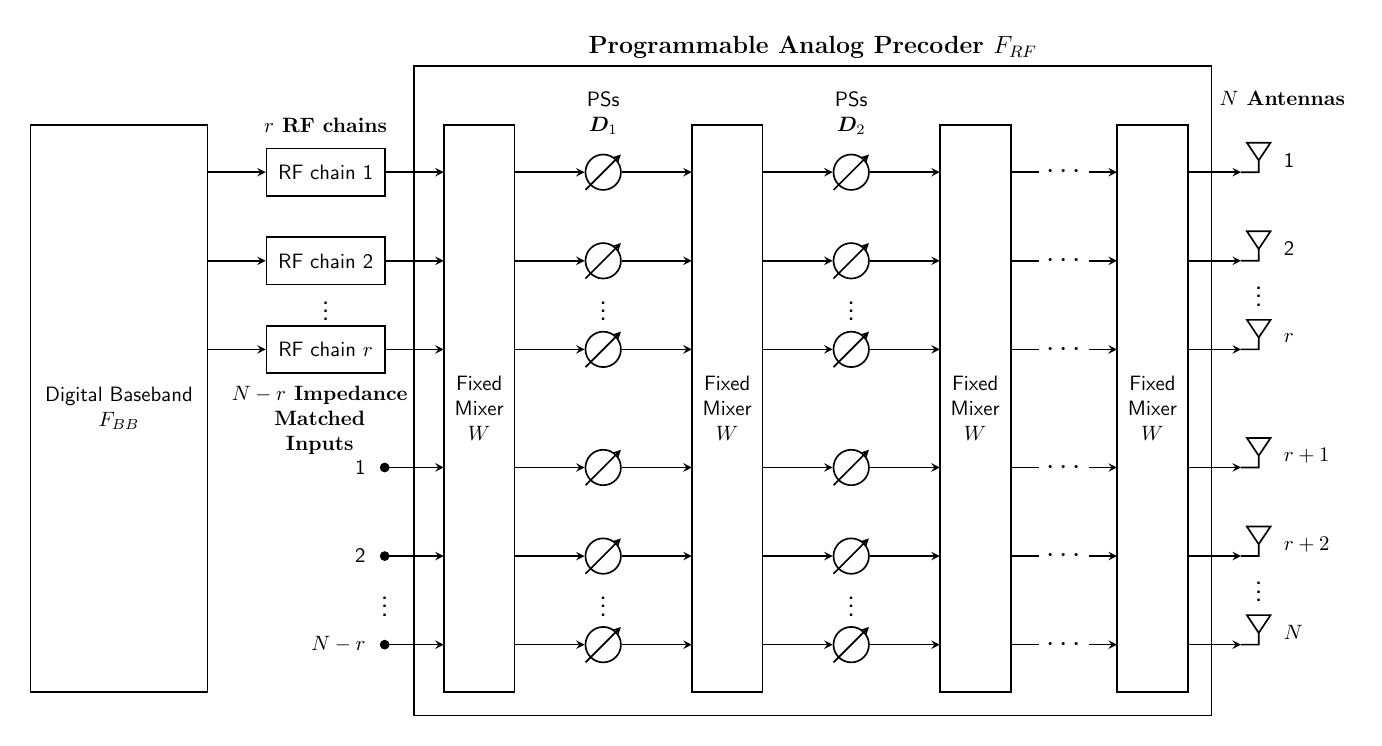}
   \caption{System diagram of the proposed hybrid beamforming architecture. The analog beamforming network has $N$ input and $N$ output ports. Its analog beamforming matrix consists of the first $r$ columns of the $N\times N$ forward transmission matrix $\mathbf X(\Phi)$.}
   \label{fig:system_diagram}
 \end{figure}
 
As can be seen from \Cref{fig:system_diagram}, the digital beamformer $\mathbf{F}_{\mathrm{BB}}$ produces the vector $\mathbf{x}$ as in \eqref{eq_tr_x}. The entries of $\mathbf{x}$ are upconverted by the $r$ RF chains and injected into the first $r$ input ports of the analog beamforming network, while the remaining $N-r$ input ports are terminated in matched loads. The vector $\mathbf{x}$ is then processed in the analog domain by subsequent mixing and phase-shifting determined by the phase-setting $\Phi$ of the analog beamformer, resulting in an antenna excitation vector
\begin{equation}\label{eq:x_ant}
\begin{aligned}
\mathbf x_{\mathrm{ant}}
&=\mathbf F_{\mathrm{RF}}(\Phi)\mathbf x
=\mathbf F_{\mathrm{RF}}(\Phi)\mathbf F_{\mathrm{BB}}\mathbf s\\
&=\mathbf X(\Phi)\mathbf E_r\mathbf F_{\mathrm{BB}}\mathbf s.
\end{aligned}
\end{equation}

\section{Analog and Digital Beamformer Design}
\label{sec:beamformer_design}
We first state the main power criterion under which all hybrid beamforming architectures in this paper will be evaluated. Then, for the proposed hybrid beamforming architecture, we jointly design the digital beamformer $\mathbf F_{\mathrm{BB}}$ and the analog beamformer $\mathbf F_{\mathrm{RF}}(\Phi)$ via its programmable phases $\Phi$ to maximize the sum rate in \eqref{eq:sumrate} under the adopted power constraint. 

\subsection{Comparison Criterion: Sum Rate Under Equal RF-Chain Injected-Power Budget}
 
In this paper, we compare the different hybrid beamforming architectures under an equal total RF-chain injected-power budget $P_T$. Hence, for every hybrid beamforming architecture evaluated in this paper, its digital beamformer must satisfy $P_{\mathrm{inj}}=\|\mathbf F_{\mathrm{BB}}\|_F^2\le P_T$. 
 
The reason is straightforward. The RF chains supply power to the analog beamforming network, which is passive and only transfers this power to the antenna ports. Even an ideal analog beamforming network can divert part of the power supplied by the RF chains to terminated ports, thereby leaving the antenna ports with less power than that supplied by the RF chains, i.e., $P_{\mathrm{rad}}\le P_{\mathrm{inj}}$ always holds. Hence, comparing hybrid beamforming architectures under equal antenna-excitation power, $P_{\mathrm{rad}}$, may compare them at different levels of power supplied by the RF chains, $P_{\mathrm{inj}}$, concealing the effect of power expenditure on the data rate. On the other hand, comparing hybrid beamforming architectures under the same power supplied by the RF chains, $P_{\mathrm{inj}}$, makes the comparison fair and exposes the effect of power expenditure on rate.

\subsection{Problem Formulation}
For a given channel $\mathbf H$ and RF-chain injected-power budget $P_T$, the sum-rate maximization problem is
\begin{equation}
\label{eq:sumrate_problem}
\begin{aligned}
\max_{\mathbf F_{\mathrm{BB}}, \Phi}\quad
&R_\Sigma(\Phi,\mathbf F_{\mathrm{BB}})\\
\text{subject to}\quad
&\|\mathbf F_{\mathrm{BB}}\|_F^2\le P_T,\\
& \mathbf F_{\mathrm{RF}}(\Phi)=\mathbf X(\Phi)\mathbf E_r
\end{aligned}
\end{equation}
where $R_\Sigma(\Phi,\mathbf F_{\mathrm{BB}})$ is given by \eqref{eq:sumrate}, $\mathbf X(\Phi)$ is given by \eqref{eq:X}, and $\Phi=\{\phi_{n,k}\}$ collects all phases from the $M$ phase-shifting layers, $\mathbf{D}_M (\bm{\phi}_M),\ldots,\mathbf{D}_1 (\bm{\phi}_1)$, in \eqref{eq:X}. 
We assume the mixing layers are fixed; optimizing the mixing matrix offline, or adapting it jointly with the programmable phase layers, is an interesting direction for future work.

Problem~\eqref{eq:sumrate_problem} is smooth, constrained, and nonconvex. We provide suboptimal solutions to \eqref{eq:sumrate_problem} using the weighted minimum mean-square error (WMMSE) reformulation of sum-rate maximization~\cite{shi2011wmmse}. To this end, we first provide several candidates for initializing $(\Phi,\mathbf F_{\mathrm{BB}})$. For each initial candidate, we update both beamformers iteratively over $t$ until the pair $(\Phi^{(t)},\mathbf F_{\mathrm{BB}}^{(t)})$ satisfies the stopping criteria or the iteration budget is reached. Finally, we select the final pair with the highest sum rate from all candidates. A brief derivation of the digital update and forward--adjoint phase-gradient expressions is provided in Appendix~\ref{app:continuous_design}.

\subsection{Generation and Selection of Initial Beamformer Candidates}
\label{subsec:capture_initialization}
We have an outer loop that generates initial candidate pairs $\big(\Phi_i^{(0)},\mathbf F_{\mathrm{BB},i}^{(0)}\big)$, where $i$ indexes the candidates.

Let $\mathbf T$ have orthonormal columns spanning $\operatorname{range}(\mathbf H)$. Generate the initial phases of the $i$-th candidate by
\begin{equation}
\label{eq:capture_initialization}
\Phi_i^{(0)}
=\operatorname{L\!\text{-}\!BFGS}_{\Phi}\!\left(
-\|\mathbf T^H\mathbf F_{\mathrm{RF}}(\Phi)\|_F^2;
\Phi_{\mathrm{seed},i}\right).
\end{equation}
Here L-BFGS denotes limited-memory quasi-Newton minimization from $\Phi_{\mathrm{seed},i}$, returning the best phase setting encountered within its iteration budget. The seed is random or supplied by DFT continuation, as specified in Appendix~\ref{app:capture_details}; the phases $\phi_{1,k}$ remain fixed at zero.

Define the effective channel by
\begin{equation}
\label{eq:effective_channel}
\mathbf H_{\mathrm{eff}}(\Phi)=\mathbf H^H\mathbf F_{\mathrm{RF}}(\Phi)
\in\mathbb C^{S\times r}.
\end{equation}
For candidate $i$, set $\mathbf H_{\mathrm{eff}}=\mathbf H_{\mathrm{eff}}(\Phi_i^{(0)})$ and compute the regularized zero-forcing directions
\begin{equation}
\label{eq:rate_mmse_initialization}
\widetilde{\mathbf F}_{\mathrm{BB},i}^{\mathrm{init}}
=\mathbf H_{\mathrm{eff}}^H
\left(\mathbf H_{\mathrm{eff}}\mathbf H_{\mathrm{eff}}^H
+\frac{S\sigma^2}{P_T}\mathbf I_S\right)^{-1}.
\end{equation}
Let $K$ be the number of nonzero columns of this matrix. For $K>0$, set
\begin{equation}
\label{eq:equal_stream_initialization}
[\mathbf F_{\mathrm{BB},i}^{(0)}]_{:,s}
=\begin{cases}
\displaystyle\sqrt{\frac{P_T}{K}}\,
\frac{[\widetilde{\mathbf F}_{\mathrm{BB},i}^{\mathrm{init}}]_{:,s}}
{\|[\widetilde{\mathbf F}_{\mathrm{BB},i}^{\mathrm{init}}]_{:,s}\|_2},
&[\widetilde{\mathbf F}_{\mathrm{BB},i}^{\mathrm{init}}]_{:,s}\ne\mathbf0,\\[1ex]
\mathbf0,&\text{otherwise}.
\end{cases}
\end{equation}

Discard candidates with $K=0$. 

In the following, each initial candidate $(\Phi_i^{(0)},\mathbf F_{\mathrm{BB},i}^{(0)})$ is further refined iteratively until the stopping criteria are met or the iteration budget is reached. For each candidate, reset $t=0$ and increment $t$ after each accepted update, continuing the same counter through fixed-phase digital optimization and joint refinement. We suppress the candidate index $i$ in the following updates. After refining all candidates, we select the final pair with the highest sum rate from all candidates.

\subsection{Digital Beamformer Design via WMMSE}
\label{subsec:digital_wmmse}
At update $t$, the next digital beamformer is given by
\begin{equation}
\label{eq:wmmse_normalize}
\mathbf F_{\mathrm{BB}}^{(t+1)}
=\sqrt{P_T}\,
\frac{\widetilde{\mathbf F}_{\mathrm{BB}}^{(t)}(\Phi^{(t+1)})}
{\|\widetilde{\mathbf F}_{\mathrm{BB}}^{(t)}(\Phi^{(t+1)})\|_F},
\end{equation}
where $\Phi^{(t+1)}$ is obtained by the analog refinement in the next subsection; for a digital-only update, set $\Phi^{(t+1)}=\Phi^{(t)}$. The normalization ensures that $\|\mathbf F_{\mathrm{BB}}^{(t+1)}\|_F^2=P_T$ at every accepted update. In \eqref{eq:wmmse_normalize}, the auxiliary matrix at any candidate phase setting $\Phi$ is given by
\begin{align}
\widetilde{\mathbf F}_{\mathrm{BB}}^{(t)}(\Phi)
={}&\bigg[\mathbf H_{\mathrm{eff}}(\Phi)^H
\operatorname{diag}\!\left(w_s^{(t)}|u_s^{(t)}|^2\right)
\mathbf H_{\mathrm{eff}}(\Phi)
\nonumber\\
&\quad+\frac{\sigma^2}{P_T}
\left(\sum_s w_s^{(t)}|u_s^{(t)}|^2\right)\mathbf I_r\bigg]^{-1}
\nonumber\\
&\times\mathbf H_{\mathrm{eff}}(\Phi)^H
\operatorname{diag}\!\left(w_s^{(t)}u_s^{(t)}\right),
\label{eq:wmmse_digital_update}
\end{align}
where $\operatorname{diag}(a_s)$ has diagonal entries $a_1,\ldots,a_S$, and $\mathbf H_{\mathrm{eff}}(\Phi)$ is the effective channel in \eqref{eq:effective_channel}.
Next, the variables $u_s^{(t)}$ and $w_s^{(t)}$ in \eqref{eq:wmmse_digital_update} are obtained as follows.
Compute the composite precoder $\mathbf F^{(t)}=\mathbf F_{\mathrm{RF}}(\Phi^{(t)})\mathbf F_{\mathrm{BB}}^{(t)}$, and let $\mathbf f_s^{(t)}$ be its $s$th column. Then, $u_s^{(t)}$ and $w_s^{(t)}$ are given by
\begin{align}
u_s^{(t)}&=\frac{\mathbf h_s^H\mathbf f_s^{(t)}}
{\sum_{j=1}^S|\mathbf h_s^H\mathbf f_j^{(t)}|^2+\sigma^2},
\label{eq:wmmse_receiver}\\
w_s^{(t)}&=\frac{\sum_{j=1}^S|\mathbf h_s^H\mathbf f_j^{(t)}|^2+\sigma^2}{\sum_{j\ne s}|\mathbf h_s^H\mathbf f_j^{(t)}|^2+\sigma^2}.
\label{eq:wmmse_weight}
\end{align}

For preliminary digital optimization, hold the phases fixed by setting $\Phi^{(t+1)}=\Phi^{(t)}$ and repeat \eqref{eq:wmmse_receiver}--\eqref{eq:wmmse_weight}, \eqref{eq:wmmse_digital_update}, and \eqref{eq:wmmse_normalize}, incrementing $t$ after each update. Stop when the digital residual and relative rate improvement meet their tolerances, or the budget is reached.

\subsection{Analog Phase Optimization with Digital Readaptation}
\label{subsec:adjoint}
At update $t$, each trial $\ell$ updates all $M$ phase-shifting layers together. For phase-shifting layer $k$, the trial phases are
\begin{equation}
\label{eq:rate_phase_update}
\bm\phi_k^{(t,\ell)}
=\bigl(\bm\phi_k^{(t)}+\eta_\ell\mathbf g_k^{(t)}\bigr)
\quad(\mathrm{mod}\ 2\pi),
\qquad \eta_\ell=\eta_0\rho^\ell,
\end{equation}
where $\eta_0>0$, $0<\rho<1$, and $\ell=0,1,\ldots$, and $\mathbf g_k^{(t)}$ is computed as
\begin{equation}
\label{eq:vec_grad}
\mathbf g_k^{(t)}
=\Im\!\left\{\left((\mathbf A_k^{(t)}\odot(\mathbf B_k^{(t)})^*)\mathbf1_S\right)
\odot e^{-j\bm\phi_k^{(t)}}\right\},
\end{equation}
where $(\cdot)^*$ is elementwise conjugation, $\odot$ is the Hadamard product, and $\mathbf B_k^{(t)}$ and $\mathbf A_k^{(t)}$ are obtained as follows.

Set $\mathbf D_k^{(t)}=\mathbf D_k(\bm\phi_k^{(t)})$. Initialize the forward computation by
\begin{equation}
\label{eq:rate_forward_initial}
\mathbf Z_0^{(t)}=\mathbf W\mathbf E_r
\widetilde{\mathbf F}_{\mathrm{BB}}^{(t)}(\Phi^{(t)}).
\end{equation}
For $k=1,\ldots,M$, compute
\begin{equation}
\label{eq:fwd_rec}
\mathbf B_k^{(t)}=\mathbf Z_{k-1}^{(t)},\qquad
\mathbf Z_k^{(t)}=\mathbf W\mathbf D_k^{(t)}\mathbf Z_{k-1}^{(t)}.
\end{equation}
Initialize the backward computation by
\begin{align}
\mathbf Y_{M+1}^{(t)}
={}&2\mathbf H\big[\operatorname{diag}(w_s^{(t)}u_s^{(t)})
-\operatorname{diag}(w_s^{(t)}|u_s^{(t)}|^2)
\mathbf H^H\mathbf Z_M^{(t)}\big].
\label{eq:reduced_wmse_adjoint_seed}
\end{align}
For $k=M,\ldots,1$, compute
\begin{equation}
\label{eq:bwd_rec}
\mathbf A_k^{(t)}=\mathbf W^H\mathbf Y_{k+1}^{(t)},\qquad
\mathbf Y_k^{(t)}=(\mathbf D_k^{(t)})^H\mathbf A_k^{(t)}.
\end{equation}

Fix $\phi_{1,k}^{(t)}=0$ and set $[\mathbf g_k^{(t)}]_1=0$. The following rule determines which trial $\ell$ in \eqref{eq:rate_phase_update} is accepted.

With $u_s^{(t)},w_s^{(t)}$ fixed, define
\begin{equation}
\label{eq:reduced_wmse}
\widehat{\mathcal J}^{(t)}(\Phi)
=
\sum_{s=1}^S\bigl(w_s^{(t)}\widetilde e_s^{(t)}-\ln w_s^{(t)}\bigr),
\end{equation}
where $\widetilde e_s^{(t)}$ is given by
\begin{align}
\widetilde e_s^{(t)}
={}&|u_s^{(t)}|^2\bigg(
\sum_{j=1}^S|[\mathbf H_{\mathrm{eff}}(\Phi)\widetilde{\mathbf F}_{\mathrm{BB}}]_{sj}|^2
+\frac{\sigma^2}{P_T}\|\widetilde{\mathbf F}_{\mathrm{BB}}\|_F^2\bigg)
\nonumber\\
&-2\Re\!\left\{(u_s^{(t)})^*[\mathbf H_{\mathrm{eff}}(\Phi)
\widetilde{\mathbf F}_{\mathrm{BB}}]_{ss}\right\}+1.
\label{eq:wmmse_virtual_mse}
\end{align}
and $\widetilde{\mathbf F}_{\mathrm{BB}}=\widetilde{\mathbf F}_{\mathrm{BB}}^{(t)}(\Phi)$ is evaluated from \eqref{eq:wmmse_digital_update} at the phase setting $\Phi$. The direction in \eqref{eq:vec_grad} is $\mathbf g_k^{(t)}=-\nabla_{\bm\phi_k}\widehat{\mathcal J}^{(t)}(\Phi)|_{\Phi=\Phi^{(t)}}$. Hold $u_s^{(t)},w_s^{(t)}$ and $\mathbf g_k^{(t)}$ fixed during the trials. For $0<c<1$, accept the first trial $\ell$ satisfying
\begin{equation}
\widehat{\mathcal J}^{(t)}(\Phi^{(t,\ell)})
\le\widehat{\mathcal J}^{(t)}(\Phi^{(t)})
-c\eta_\ell\sum_{k=1}^M\|\mathbf g_k^{(t)}\|_2^2
\label{eq:rate_armijo}
\end{equation}
and then set $\Phi^{(t+1)}=\Phi^{(t,\ell)}$. Each trial uses its own digital auxiliary matrix from \eqref{eq:wmmse_digital_update}. If all phase gradients vanish, keep $\Phi^{(t+1)}=\Phi^{(t)}$ and perform the digital update. Recompute $u_s^{(t+1)},w_s^{(t+1)}$ from the accepted normalized pair at the next iteration.

\subsection{Stopping Criteria, Overall Algorithm, and Computational Complexity}
\label{subsec:rate_programming}
Write $R_\Sigma^{(t)}=R_\Sigma(\Phi^{(t)},\mathbf F_{\mathrm{BB}}^{(t)})$. We use three stopping criteria: i) the relative rate improvement $(R_\Sigma^{(t+1)}-R_\Sigma^{(t)})/R_\Sigma^{(t)}$; ii) the digital-update residual 
 \begin{equation}
\frac{\|\widetilde{\mathbf F}_{\mathrm{BB}}^{(t)}(\Phi^{(t)})-\mathbf F_{\mathrm{BB}}^{(t)}\|_F}
{\sqrt{P_T}};
\label{eq:digital_update_residual}
\end{equation}
and iii) the phase-gradient norm. The residual and gradient norm are evaluated before the accepted update, and the rate improvement afterward. Refinement stops when all three meet their tolerances, or its iteration budget is reached.
 
An overview of the complete algorithm is given in Algorithm~\ref{alg:adjoint_prog}.

\begin{algorithm}[!ht]
\caption{Continuous-phase hybrid beamformer design}
\label{alg:adjoint_prog}
\begin{algorithmic}[1]
\STATE \textbf{Input:} $\mathbf H$, $\sigma^2$, $P_T$, $r$, $M$, $\mathbf W$, candidate seeds $\{\Phi_{\mathrm{seed},i}\}$, budgets and tolerances, and $\eta_0,c,\rho$.
\STATE Form $\mathbf T$ with orthonormal columns spanning $\operatorname{range}(\mathbf H)$.
\FOR{each candidate $i$}
\STATE Generate $\Phi_i^{(0)}$ by \eqref{eq:capture_initialization}, using Appendix~\ref{app:capture_details} for continuation seeds.
\STATE Form $\widetilde{\mathbf F}_{\mathrm{BB},i}^{\mathrm{init}}$ by \eqref{eq:rate_mmse_initialization}; discard the candidate if all columns vanish. Otherwise form $\mathbf F_{\mathrm{BB},i}^{(0)}$ by \eqref{eq:equal_stream_initialization}.
\STATE Set $t=0$, $\Phi^{(0)}=\Phi_i^{(0)}$, and $\mathbf F_{\mathrm{BB}}^{(0)}=\mathbf F_{\mathrm{BB},i}^{(0)}$ for this run.
\STATE Apply the fixed-phase digital updates of Subsection~\ref{subsec:digital_wmmse}, incrementing $t$ after each update; use the safeguard of Appendix~\ref{app:digital_safeguard} if needed.
\REPEAT
\STATE Compute $u_s^{(t)},w_s^{(t)}$ from the current pair using \eqref{eq:wmmse_receiver}--\eqref{eq:wmmse_weight}.
\STATE Compute $\widetilde{\mathbf F}_{\mathrm{BB}}^{(t)}(\Phi^{(t)})$ and $\widehat{\mathcal J}^{(t)}(\Phi^{(t)})$ by \eqref{eq:wmmse_digital_update} and \eqref{eq:reduced_wmse}; record \eqref{eq:digital_update_residual}.
\STATE Compute \eqref{eq:rate_forward_initial}, \eqref{eq:fwd_rec}, \eqref{eq:reduced_wmse_adjoint_seed}, \eqref{eq:bwd_rec}, and \eqref{eq:vec_grad}; set $[\mathbf g_k^{(t)}]_1=0$.
\STATE Keeping $u_s^{(t)},w_s^{(t)}$ fixed, test \eqref{eq:rate_phase_update} for $\ell=0,1,\ldots$, computing the digital minimizer and reduced objective at each trial. Set $\Phi^{(t+1)}$ to the first trial satisfying \eqref{eq:rate_armijo}; retain $\Phi^{(t)}$ if all phase gradients vanish.
\STATE Compute $\mathbf F_{\mathrm{BB}}^{(t+1)}$ by \eqref{eq:wmmse_normalize}; evaluate $R_\Sigma^{(t+1)}$ and the relative rate improvement; set $t\leftarrow t+1$.
\UNTIL{the residual, gradient, and rate tolerances are met, or the refinement budget is reached}
\STATE Retain the final pair and its rate.
\ENDFOR
\STATE \textbf{Output:} The highest-rate pair; retain all final pairs for phase quantization.
\end{algorithmic}
\end{algorithm}

\paragraph{Computational complexity}
For DFT mixing, one capture objective/gradient evaluation costs $\mathcal O(MNr\log N+NrS)$. At fixed phases, one digital WMMSE update costs $\mathcal O(r^3+r^2S+rS^2)$, with no further analog-layer propagation. One joint-refinement iteration, including up to $L_{\mathrm{bt}}$ line-search trials, costs

\begin{equation}
\label{eq:programming_complexity}
\mathcal O\!\left((1+L_{\mathrm{bt}})
\big[MNS\log N+r^3+r^2S+rS^2\big]+NS^2\right).
\end{equation}

Total work sums these costs, including trial digital solves, over iterations, candidates, and power settings. Forward states require $\mathcal O(MNr)$ storage for capture and $\mathcal O(MNS)$ for refinement, in separate stages. FFT evaluation avoids storing a dense DFT matrix; for dense unstructured mixing, replace $N\log N$ by $N^2$.

\section{Fully Digital Sum-Rate Attainability, Phase-Layer Requirements, and Rate Gap}
\label{sec:depth}
For the proposed hybrid beamforming architecture under ideal passive-component assumptions, in this section, we characterize fully-digital sum-rate attainment, establish necessary and sufficient phase-layer conditions for such attainment, and bound the gap between its best attainable sum rate and the fully-digital sum-rate optimum.

\subsection{Fully Digital Sum-Rate Attainment}
\label{subsec:rate_attainment}
For a fixed channel $\mathbf H$ and noise variance $\sigma^2$, write $R_\Sigma(\mathbf F)$ for \eqref{eq:sumrate} evaluated at the composite precoder $\mathbf F$. Under the common RF-chain injected-power budget $P_T$, define
\begin{align}
R_{\mathrm{FD}}^\star(\mathbf H)
&=\max_{\|\mathbf F\|_F^2\le P_T}R_\Sigma(\mathbf F),
\label{eq:fd_optimal_rate}\\
R_{\mathrm{HYB}}^\star(\mathbf H;M)
&=\max_{\substack{\Phi,\,\mathbf F_{\mathrm{BB}}\\
\|\mathbf F_{\mathrm{BB}}\|_F^2\le P_T}}
R_\Sigma\!\left(\mathbf F_{\mathrm{RF}}(\Phi)\mathbf F_{\mathrm{BB}}\right),
\label{eq:hybrid_optimal_rate}
\end{align}
where the analog beamformer has $M$ programmable phase layers. Semi-unitarity gives $\|\mathbf F_{\mathrm{RF}}\mathbf F_{\mathrm{BB}}\|_F^2=\|\mathbf F_{\mathrm{BB}}\|_F^2$, so every feasible hybrid precoder is fully-digital feasible and
\begin{equation}
R_{\mathrm{HYB}}^\star(\mathbf H;M)
\le R_{\mathrm{FD}}^\star(\mathbf H).
\label{eq:fd_rate_upper_bound}
\end{equation}

\begin{proposition}[Fully digital sum-rate attainment under channel-subspace containment]
\label{prop:rate_containment}
If a realizable phase setting satisfies
\begin{equation}
\label{eq:cont_2}
\operatorname{range}(\mathbf H)
\subseteq\operatorname{range}\!\left(\mathbf F_{\mathrm{RF}}(\Phi)\right),
\end{equation}
then, under the same RF-chain injected-power budget,
\begin{equation}
R_{\mathrm{HYB}}^\star(\mathbf H;M)
=R_{\mathrm{FD}}^\star(\mathbf H).
\label{eq:fd_rate_attainment}
\end{equation}
\end{proposition}
\begin{IEEEproof}
See Appendix~\ref{app:containment}.
\end{IEEEproof}
Condition~\eqref{eq:cont_2} allows the digital beamformer to preserve all desired and interfering received signals of a fully-digital sum-rate-optimal precoder. With $\mathbf T$ an orthonormal basis of $\operatorname{range}(\mathbf H)$, \eqref{eq:cont_2} is equivalent to
\begin{equation}
\label{eq:capture_containment}
\|\mathbf T^H\mathbf F_{\mathrm{RF}}(\Phi)\|_F^2
=\operatorname{rank}(\mathbf H),
\end{equation}
since each squared projection norm in the sum is at most one. This motivates the candidate generation in Subsection~\ref{subsec:capture_initialization}, which seeks an analog subspace that supports fully-digital sum-rate attainment, before digital optimization and joint rate refinement.

\subsection{Necessary and Sufficient Phase-Layer Conditions}
\label{subsec:necessary_layers}
Universal attainment means attaining the fully-digital sum-rate optimum for every full-column-rank channel, with the beamformers selected separately for each channel.

\begin{proposition}[Necessary phase-layer count and generic high-SNR rate loss]
\label{prop:stream_dof}
Let $1\le S\le r<N$, with fixed unitary $\mathbf W$, continuous phases, and ideal passive components. Universal fully-digital sum-rate attainment at fixed $P_T$ and $\sigma^2$ requires
\begin{equation}
\label{eq:stream_dof_bound}
M\ge\left\lceil\frac{2S(N-r)}{N-1}\right\rceil.
\end{equation}
If this bound fails, then for almost every full-column-rank $\mathbf H\in\mathbb C^{N\times S}$ no realizable phase setting satisfies \eqref{eq:cont_2}. For each such channel, there exists a finite $\rho_0(\mathbf H)$ such that
\[
R_{\mathrm{HYB}}^\star(\mathbf H;M)<R_{\mathrm{FD}}^\star(\mathbf H)
\quad\text{whenever }P_T/\sigma^2>\rho_0(\mathbf H).
\]
\end{proposition}
\begin{IEEEproof}
See Appendix~\ref{app:stream_dof}.
\end{IEEEproof}
Here ``almost every'' refers to Lebesgue measure on the full-rank channel matrices, and hence to channel laws with a density on that space. The necessary count is approximately $2S$ when $r=S\ll N$. Increasing $r$ lowers or leaves unchanged this count because a larger analog subspace can contain more channel subspaces.

\begin{proposition}[Sufficient phase-layer count for an open set of channels]
\label{prop:local_sufficiency}
Let $\mathbf W$ be the normalized $N$-point DFT matrix, with continuous phases, ideal passive components, and $1\le S\le r<N$. If $r\ge S+1$ or $S=1$, then for every
\begin{equation}
M\ge2S
\label{eq:local_sufficient_layers}
\end{equation}
there is a nonempty open set of full-column-rank channel matrices for which \eqref{eq:cont_2} is realizable. Every channel in this set therefore satisfies \eqref{eq:fd_rate_attainment} for all power budgets $P_T$ and noise variances $\sigma^2$.
\end{proposition}
\begin{IEEEproof}
See Appendix~\ref{app:local_sufficiency}.
\end{IEEEproof}
This establishes exact fully-digital sum-rate attainment with the proposed hybrid beamforming architecture for the stated channel set and RF-chain conditions.

\begin{remark}
Prior results establish realization of arbitrary unitary $N\times N$ matrices using DFT mixing and phase layers~\cite{lopezpastor2021arbitrary,girouard2026nearoptimal}. For even $N$, the construction in~\cite{girouard2026nearoptimal} gives the universal sufficient count $M=2N+5$ in the convention of \eqref{eq:X}. Our condition \eqref{eq:local_sufficient_layers} gives a smaller sufficient count for the open channel set and RF-chain conditions in Proposition~\ref{prop:local_sufficiency}. This is because Proposition~\ref{prop:local_sufficiency} addresses the beamforming task of attaining \eqref{eq:fd_rate_attainment} on that channel set, and not the full-unitary matrix constructions addressed in \cite{lopezpastor2021arbitrary,girouard2026nearoptimal}.
\end{remark}

\subsection{Rate Gap at a Finite Number of Phase Layers}
\label{subsec:finite_layer_gap}

Proposition~\ref{prop:local_sufficiency} establishes exact attainment for the stated channel set. For a general channel and a finite number of phase layers, we next bound the gap between the fully-digital sum-rate optimum and the best sum rate attainable by the proposed hybrid beamforming architecture through its channel-subspace capture.

For full-column-rank $\mathbf H$, let $\mathbf T$ contain an orthonormal basis of its $S$-dimensional column space. Define the unnormalized capture deficit and largest principal angle by
\begin{align}
\delta(\Phi)&=S-\|\mathbf T^H\mathbf F_{\mathrm{RF}}(\Phi)\|_F^2,
\label{eq:capture_deficit}\\
\cos\psi_{\max}(\Phi)&=\sigma_{\min}\!\left(\mathbf T^H\mathbf F_{\mathrm{RF}}(\Phi)\right),
\label{eq:max_principal_angle}
\end{align}
where $\sigma_{\min}$ is the smallest of the $S$ singular values.

\begin{proposition}[Sum-rate gap from channel-subspace capture]
\label{prop:rate_gap}
For any realizable $M$-layer phase setting with $\psi_{\max}(\Phi)<\pi/2$,
\begin{equation}
\begin{split}
0&\le R_{\mathrm{FD}}^\star(\mathbf H)-R_{\mathrm{HYB}}^\star(\mathbf H;M)\\
&\le S\log_2\!\left(\frac{1}{\cos^2\psi_{\max}(\Phi)}\right).
\end{split}
\label{eq:principal_angle_rate_gap}
\end{equation}
If $\delta(\Phi)<1$, this also gives
\begin{equation}
R_{\mathrm{FD}}^\star(\mathbf H)-R_{\mathrm{HYB}}^\star(\mathbf H;M)
\le S\log_2\!\left(\frac{1}{1-\delta(\Phi)}\right).
\label{eq:capture_rate_gap}
\end{equation}
Both bounds hold for every power budget $P_T$ and noise variance $\sigma^2$.
\end{proposition}
\begin{IEEEproof}
See Appendix~\ref{app:rate_gap}.
\end{IEEEproof}
A small capture deficit therefore gives a sufficient guarantee of a small achievable sum-rate gap uniformly in SNR. The same bounds apply to quantized phases when $\Phi$ belongs to the phase alphabet and \eqref{eq:hybrid_optimal_rate} is optimized over that alphabet, since semi-unitarity is preserved.

\section{Analog Beamformer Design with Quantized Phase Layers}
\label{sec:quantized}
We maximize the sum rate in \eqref{eq:sumrate_problem} when the analog phases are quantized. Starting from the completed continuous-phase designs of Section~\ref{sec:beamformer_design}, we quantize their phases and adapt their digital beamformers to the resulting effective channels. We then jointly refine the quantized phases and digital beamformers, and select the final pair with the highest sum rate.

\subsection{Quantized Phase Layers and Problem Formulation}
\label{subsec:quantized_layers}
Following \cite{keshavarz2025programmable}, each phase element is a $q$-bit variable, with $q\ge1$,
\begin{equation}
\label{eq:phi_quant}
\phi_{n,k}\in\mathcal Q_q\triangleq
\{0,\Delta,2\Delta,\ldots,(2^q-1)\Delta\},
\qquad \Delta=\frac{2\pi}{2^q}.
\end{equation}

For the quantized case, the design problem of the digital and analog beamformers is formulated identically to the continuous-phase case, given by \eqref{eq:sumrate_problem}, but now with the additional constraint that all phases must satisfy \eqref{eq:phi_quant}, i.e., the phase shifts are quantized.

We note that quantized diagonal layers remain unitary, so Proposition~\ref{prop:power_preserve} still applies under ideal passive-component assumptions. We maintain $\|\mathbf F_{\mathrm{BB}}\|_F^2=P_T$. Moreover, we fix $\phi_{1,k}=0$, since subtracting a common grid phase from a layer preserves $\mathcal Q_q$ and changes the composite precoder only by a common phase.

\subsection{Initial Quantized Phase Settings and Digital Beamformers}
\label{subsec:quantized_digital}
For every candidate $i$, we take the final optimized analog--digital pair $(\Phi_i, \mathbf F_{\mathrm{BB},i})$ obtained from Algorithm~\ref{alg:adjoint_prog} as a starting point for the quantized-phase design. Keeping several pairs allows their rates after quantization and refinement to determine the best final analog--digital beamformer pair. 

We now drop the candidate index $i$ and explain how a given optimized analog--digital pair $(\Phi, \mathbf F_{\mathrm{BB}})$ obtained from Algorithm~\ref{alg:adjoint_prog} is refined for the quantized case.

First, initialize the quantized phases by
\begin{equation}
\Phi^{(0)}=Q_q(\Phi),
\label{eq:quantized_initialization}
\end{equation}
where $Q_q(\cdot)$ denotes elementwise nearest-grid quantization modulo $2\pi$. 
Next, recompute $\mathbf H_{\mathrm{eff}}(\Phi^{(0)})$ using \eqref{eq:effective_channel}. Since quantization in \eqref{eq:quantized_initialization} may change the effective channel, we adapt the digital beamformer $\mathbf F_{\mathrm{BB}}$ to $\mathbf H_{\mathrm{eff}}(\Phi^{(0)})$ while holding $\Phi^{(0)}$ fixed. We use two digital starts: the inherited full-power beamformer from the corresponding continuous-phase pair, and a fresh initialization from \eqref{eq:rate_mmse_initialization}--\eqref{eq:equal_stream_initialization}, evaluated on the quantized effective channel. Next, we
refine each digital start separately using Subsection~\ref{subsec:digital_wmmse}, which recomputes $u_s^{(t)}$ and $w_s^{(t)}$, and applies \eqref{eq:wmmse_digital_update} and \eqref{eq:wmmse_normalize}. Finally, we retain the refined digital beamformer with the higher sum rate of the two. The resulting pair supplies the initial quantized phases and adapted digital beamformer for the joint refinement below. Candidate handling and iteration-counter details are given in Appendix~\ref{app:quantized_candidates}.

\subsection{Joint Refinement of Quantized Analog and Digital Beamformers}
\label{subsec:quantized_design}
At the beginning of a phase sweep, compute $u_s^{(t)},w_s^{(t)}$ from the current full-power pair $(\Phi^{(t)},\mathbf F_{\mathrm{BB}}^{(t)})$ with positive sum rate using \eqref{eq:wmmse_receiver}--\eqref{eq:wmmse_weight}. Hold these quantities fixed while updating the phases sequentially over $k=1,\ldots,M$ and $n=2,\ldots,N$, using the latest accepted values of all other phases.

For one coordinate $\phi=\phi_{n,k}$, let $\widehat{\mathcal J}(\phi)$ denote \eqref{eq:reduced_wmse} evaluated at the trial phase setting. Each trial uses its own digital minimizer $\widetilde{\mathbf F}_{\mathrm{BB}}^{(t)}(\Phi)$ from \eqref{eq:wmmse_digital_update}, computed with the fixed $u_s^{(t)},w_s^{(t)}$. Thus, the digital minimizer is readapted to each trial effective channel through the same closed-form expression used in the continuous-phase design. We suppress the fixed sweep index $t$ in the following scalar expressions.

\begin{proposition}[Quantized-phase selection with digital readaptation]
\label{prop:quantized_phase_update}
Fix $u_s^{(t)},w_s^{(t)}$ obtained as above, and all phases except $\phi_{n,k}$. For a nonconstant reduced objective, an exact minimizer over $\mathcal Q_q$ is
\begin{equation}
\phi_{n,k}^{\mathrm{new}}\in
\underset{\phi\in\{\phi_-,\phi_+\}}{\operatorname{arg\,min}}\;
\widehat{\mathcal J}(\phi),
\label{eq:quantized_phase_update}
\end{equation}
where
\begin{equation}
\phi_-=\Delta\left\lfloor\frac{\phi_\star}{\Delta}\right\rfloor,\qquad
\phi_+=\Delta\left\lceil\frac{\phi_\star}{\Delta}\right\rceil
\quad(\mathrm{mod}\ 2\pi).
\label{eq:quantized_grid_neighbors}
\end{equation}
Here $\phi_\star$ is the closed-form continuous-phase minimizer in \eqref{eq:quantized_continuous_minimum}, and the scalar reduced objective is
\begin{equation}
\widehat{\mathcal J}(\phi)=
\frac{a_0+a_{\mathrm c}\cos\phi+a_{\mathrm s}\sin\phi}
{b_0+b_{\mathrm c}\cos\phi+b_{\mathrm s}\sin\phi},
\label{eq:quantized_reduced_scalar}
\end{equation}
with real coefficients and $b_0>\sqrt{b_{\mathrm c}^2+b_{\mathrm s}^2}$. The coefficients are obtained by \eqref{eq:quantized_coefficient_samples}--\eqref{eq:quantized_coefficient_recovery}. Retain the current phase for a constant objective or when it is tied for the minimum.
\end{proposition}
\begin{IEEEproof}
See Appendix~\ref{app:quantized_reduced}.
\end{IEEEproof}

The scalar minimizer $\phi_\star$ is recomputed for each coordinate using the current values of the other phases. Selecting between its two neighboring grid values allows further refinement beyond the initial rounding in \eqref{eq:quantized_initialization}. Trial digital readaptation evaluates \eqref{eq:wmmse_digital_update}; repeated fixed-phase digital iterations follow the completed sweep.

After visiting all programmable phases, set $\Phi^{(t+1)}$ to the completed setting and compute $\mathbf F_{\mathrm{BB}}^{(t+1)}$ by \eqref{eq:wmmse_normalize}, using the sweep's fixed $u_s^{(t)},w_s^{(t)}$ in \eqref{eq:wmmse_digital_update}. This ensures $\|\mathbf F_{\mathrm{BB}}^{(t+1)}\|_F^2=P_T$ and completes one update. Next, hold these phases fixed and refine the digital beamformer by Subsection~\ref{subsec:digital_wmmse}, recomputing $u_s^{(t)},w_s^{(t)}$ and incrementing $t$ after each accepted update. The next phase sweep starts from this updated pair.

\subsection{Overall Algorithm, Stopping Criteria, and Computational Complexity}
\label{subsec:quantized_algorithm}
For each candidate, repeat phase sweeps and digital refinement until a sweep leaves all phases unchanged and the digital residual/rate criteria hold, or the sweep budget is reached. Digital refinement is also performed after the final allowed sweep. Select the highest-rate final pair across candidates: its phase setting is the refined quantized $\Phi$, and its digital beamformer is the associated adapted $\mathbf F_{\mathrm{BB}}$. Algorithm~\ref{alg:quantized_rate} summarizes the procedure.

\begin{algorithm}[t]
\caption{Hybrid beamformer design with quantized phase layers}
\label{alg:quantized_rate}
\begin{algorithmic}[1]
\STATE \textbf{Input:} $\mathbf H,\sigma^2,P_T,r,M,\mathbf W,q$, final continuous candidate pairs from Algorithm~\ref{alg:adjoint_prog}, tolerances, and budgets.
\STATE Quantize each retained phase setting by \eqref{eq:quantized_initialization}; handle duplicate settings as in Appendix~\ref{app:quantized_candidates}.
\FOR{each distinct quantized setting}
\STATE Recompute its effective channel by \eqref{eq:effective_channel}.
\STATE For a zero effective channel, retain the quantized phase setting with the inherited digital beamformer and proceed to the next candidate.
\STATE Refine the inherited and fresh RZF digital starts separately at fixed phases; retain the higher-rate refined pair and its counter $t$.
\REPEAT
\STATE Compute $u_s^{(t)},w_s^{(t)}$ by \eqref{eq:wmmse_receiver}--\eqref{eq:wmmse_weight}; hold them fixed for this sweep.
\FOR{$k=1,\ldots,M$ and, within each layer, $n=2,\ldots,N$}
\STATE Update $\phi_{n,k}$ by \eqref{eq:quantized_phase_update}, readapting the digital minimizer to each trial setting by \eqref{eq:wmmse_digital_update}.
\ENDFOR
\STATE Set $\Phi^{(t+1)}$ to the completed setting and normalize its digital minimizer by \eqref{eq:wmmse_normalize}; set $t\leftarrow t+1$.
\STATE Refine the digital beamformer at fixed phases by Subsection~\ref{subsec:digital_wmmse}, continuing $t$.
\UNTIL{an unchanged sweep and the digital residual/rate criteria hold, or the sweep budget is reached}
\STATE Retain the final quantized analog--digital pair and its sum rate.
\ENDFOR
\STATE \textbf{Output:} The retained pair with the highest sum rate.
\end{algorithmic}
\end{algorithm}

Completed sweeps and digital refinement preserve full RF-chain injected power and give nondecreasing rates bounded by $R_{\mathrm{FD}}^\star(\mathbf H)$, so the rate sequence converges. Appendix~\ref{app:quantized_reduced} gives the sweep evaluation and numerical safeguards.

\paragraph{Computational complexity}
For DFT mixing, a phase sweep costs
\begin{equation}
\mathcal O\!\left(MN\big[(r+S)\log N+r^2+rS\big]\right).
\label{eq:quantized_sweep_complexity}
\end{equation}
Each coordinate uses three coefficient evaluations and at most two grid evaluations, independently of $q$. Combine this with the digital costs in Subsection~\ref{subsec:rate_programming}; Appendix~\ref{app:quantized_cost} outlines the factorization reuse. For dense mixing, replace $N\log N$ by $N^2$.

\section{Comparator Architectures Under Equal RF-Chain Injected Power}
\label{sec:baselines}
This section specifies the fully-digital and the four hybrid beamforming architectures that we use as baselines in this paper. All architectures use the same channel, receiver-noise variance, and total RF-chain injected-power budget. We first state their common sum-rate maximization problem and then define the physical analog beamformer and its adjustable parameters for each architecture.

\subsection{Common Sum-Rate Maximization Problem}
\label{subsec:comparator_optimization}
For each baseline architecture, the design problem is
\begin{equation}
\begin{aligned}
\max_{\mathbf F_{\mathrm{BB}},\,\mathbf F_{\mathrm{RF}}}\quad
&R_\Sigma\!\left(\mathbf F_{\mathrm{RF}}\mathbf F_{\mathrm{BB}}\right)\\
\text{subject to}\quad
&\|\mathbf F_{\mathrm{BB}}\|_F^2\le P_T,\\
&\mathbf F_{\mathrm{RF}}\in\mathcal F_{\mathrm{RF}},
\end{aligned}
\label{eq:comparator_sumrate_problem}
\end{equation}
where $R_\Sigma(\mathbf F)$ is the sum rate in \eqref{eq:sumrate} evaluated at the composite precoder $\mathbf F$, and $\mathcal F_{\mathrm{RF}}$ is the set of physical analog beamformers realizable by the considered architecture. These sets are specified below through the analog matrices and the allowed values of their control parameters. The digital beamformer has $S$ columns and one row per RF chain.

\subsection{Fully Digital Sum-Rate Optimum}
\label{subsec:fd_baseline}
The fully-digital transmitter drives the antennas directly, so $\mathcal F_{\mathrm{RF}}=\{\mathbf I_N\}$ and only the digital beamformer $\mathbf F_{\mathrm{BB}}\in\mathbb C^{N\times S}$ is optimized in \eqref{eq:comparator_sumrate_problem}. By design, the fully-digital transmitter uses $N$ RF chains.

\subsection{Sohrabi--Yu Fully-Connected Architecture}
\label{subsec:fc1_baseline}
The Sohrabi--Yu architecture (FC1)~\cite{sohrabi2016hybrid} uses one phase shifter per RF-chain-to-antenna connection. Its ideal passive analog beamforming network splits each RF-chain signal into $N$ branches and combines $r_{\mathrm{FC1}}$ branches at each antenna port. By design, this architecture uses $r_{\mathrm{FC1}}=2S$ RF chains. The physical analog beamformer is
\begin{equation}
\label{eq:fc1_phys}
[\mathbf F_{\mathrm{RF}}^{\mathrm{FC1}}(\Phi)]_{n\ell}
=\frac{e^{j\phi_{n\ell}}}{\sqrt{N r_{\mathrm{FC1}}}},
\end{equation}
where $\Phi=\{\phi_{n\ell}\}$ contains the phases for $n=1,\ldots,N$ and $\ell=1,\ldots,r_{\mathrm{FC1}}$, with $\phi_{n\ell}\in[0,2\pi)$. Varying these phases defines $\mathcal F_{\mathrm{RF}}$ for FC1. The factors $1/\sqrt N$ and $1/\sqrt{r_{\mathrm{FC1}}}$ represent ideal splitting and combining, respectively. Since $\|\mathbf F_{\mathrm{RF}}^{\mathrm{FC1}}\|_F=1$, this transfer is contractive.

\subsection{Yu--Zhang--Letaief Fully-Connected Architecture}
\label{subsec:fc2_baseline}
The Yu--Zhang--Letaief architecture (FC2)~\cite{yu2019dps} uses two phase shifters per RF-chain-to-antenna connection. This architecture uses $r$ RF chains, with $r\ge S$. Splitting each RF-chain signal into $2N$ branches and combining $2r$ branches at each antenna gives
\begin{equation}
\label{eq:fc2_phys}
[\mathbf F_{\mathrm{RF}}^{\mathrm{FC2}}(\Phi)]_{n\ell}
=\frac{e^{j\phi_{n\ell}^{(1)}}+e^{j\phi_{n\ell}^{(2)}}}{2\sqrt{Nr}}.
\end{equation}
Here $\Phi=\{\phi_{n\ell}^{(1)},\phi_{n\ell}^{(2)}\}$ contains both phase arrays for $n=1,\ldots,N$ and $\ell=1,\ldots,r$, with every phase in $[0,2\pi)$. These independently adjustable phases define $\mathcal F_{\mathrm{RF}}$ for FC2. Every entry has magnitude at most $1/\sqrt{Nr}$, so $\|\mathbf F_{\mathrm{RF}}^{\mathrm{FC2}}\|_F\le1$ and the transfer is contractive.

\subsection{Partially Connected Architecture}
\label{subsec:pc_baseline}
The partially connected architecture assigns one RF chain to each disjoint antenna subarray~\cite{gao2016ee}. This architecture uses $r$ RF chains, with $r\ge S$. For the considered settings, $r$ divides $N$; we use $r$ fixed contiguous subarrays $\mathcal I_\ell$, each containing $N/r$ antennas. The physical analog beamformer is
\begin{equation}
[\mathbf F_{\mathrm{RF}}^{\mathrm{PC}}(\Phi)]_{n\ell}
=\begin{cases}
 e^{j\phi_{n\ell}}/\sqrt{N/r},&n\in\mathcal I_\ell,\\
 0,&n\notin\mathcal I_\ell.
\end{cases}
\label{eq:pc_phys}
\end{equation}
The control variables are $\Phi=\{\phi_{n\ell}:n\in\mathcal I_\ell,\ \ell=1,\ldots,r\}$, with every phase in $[0,2\pi)$. Their variation defines $\mathcal F_{\mathrm{RF}}$, while the subarray connections remain fixed. The columns have disjoint support and unit norm, giving $(\mathbf F_{\mathrm{RF}}^{\mathrm{PC}})^H\mathbf F_{\mathrm{RF}}^{\mathrm{PC}}=\mathbf I_r$ and $P_{\mathrm{rad}}=P_{\mathrm{inj}}$.

\subsection{Butler/DFT Beam-Selection Architecture}
\label{subsec:butler_baseline}
The Butler/DFT comparator uses an ideal fixed Fourier analog beamforming network inspired by~\cite{garcia2016rf}. This architecture uses $r$ RF chains, with $r\ge S$. Let $\mathbf U_{\mathrm{DFT}}$ be a unitary $N\times N$ DFT matrix with columns $\mathbf u_n$. Its control variable is the set $\mathcal B\subseteq\{1,\ldots,N\}$ of $r$ distinct selected beam indices. The physical analog beamformer is
\begin{equation}
\mathbf F_{\mathrm{RF}}^{\mathrm{BM}}(\mathcal B)
=\mathbf U_{\mathrm{DFT}}\mathbf E_{\mathcal B},
\qquad |\mathcal B|=r,
\label{eq:butler_phys}
\end{equation}
where $\mathbf E_{\mathcal B}$ contains the columns of $\mathbf I_N$ indexed by $\mathcal B$. All such beam selections define $\mathcal F_{\mathrm{RF}}$ for this architecture. The selected columns remain orthonormal, giving $(\mathbf F_{\mathrm{RF}}^{\mathrm{BM}})^H\mathbf F_{\mathrm{RF}}^{\mathrm{BM}}=\mathbf I_r$ and $P_{\mathrm{rad}}=P_{\mathrm{inj}}$.

\section{Numerical Results}
\label{sec:numerics}
In this section, we evaluate the proposed hybrid beamforming architecture in a narrowband multiuser downlink scenario, first under ideal passive-component assumptions and then with component losses. The comparisons include the fully-digital sum-rate optimum, FC1, FC2, the partially connected architecture, and Butler/DFT beam selection, as defined in Section~\ref{sec:baselines}. We examine the number of programmable phase layers, phase resolution, RF-chain count, the power-normalization criterion, and component attenuation.

\subsection{System and Channel Setup}
We consider a ULA with $N=512$ antennas at $f_c=100$~GHz, half-wavelength spacing $d=\lambda_c/2$, and $S=16$ single-antenna users. The proposed architecture uses $r\in\{16,32\}$ RF chains and a fixed unitary DFT mixing matrix $\mathbf W$. FC2, the partially connected architecture, and Butler/DFT beam selection use the same $r$; FC1 uses $2S=32$ RF chains throughout. The proposed phase shifters have continuous control or $q\in\{2,4,6\}$ bits, while the other programmable analog beamformers use continuous phases. 

User $k$ has distance $\rho_k\sim\mathcal U[50,1000]$~m and angle $\theta_k\sim\mathcal U[-\pi/3,\pi/3]$, measured from the array axis. With antenna position $x_n=(n-1-(N-1)/2)d$, define the spherical-wave path length
\begin{equation}
R_n(\rho,\theta)=\sqrt{\rho^2+x_n^2-2\rho x_n\cos\theta}.
\end{equation}
The channel contains one LOS path and $L=4$ NLOS paths:
\begin{align}
[\mathbf h_k]_n
&=\frac{\lambda_c}{4\pi\rho_k}
 e^{-j\frac{2\pi}{\lambda_c}R_n(\rho_k,\theta_k)}\nonumber\\
&\quad+\sum_{p=1}^{L}\frac{\lambda_c}{4\pi\rho_{k,p}}\alpha_{k,p}
 e^{-j\frac{2\pi}{\lambda_c}R_n(\rho_{k,p},\theta_{k,p})},
\label{eq:sim_channel}
\end{align}
where the NLOS path distances and angles are independent draws from the same ranges, and $\alpha_{k,p}=10^{-15/20}e^{j\varphi_{k,p}}$ with $\varphi_{k,p}\sim\mathcal U[0,2\pi)$. This model retains distance-dependent attenuation and spherical-wave phase curvature. The noise density is $-174$~dBm/Hz over $B=200$~kHz, giving $\sigma^2=7.96\times10^{-16}$~W without an additional receiver noise figure.

The beamformers follow Algorithms~\ref{alg:adjoint_prog} and~\ref{alg:quantized_rate} and the comparator designs in Section~\ref{sec:baselines}, with rates evaluated by \eqref{eq:sim_sinr}--\eqref{eq:sumrate}. The comparator problems are evaluated numerically using digital WMMSE updates and alternating refinement of the analog controls, retaining the highest-rate solution across the initializations. Analog refinement uses phase-coordinate updates for the fully-connected and partially connected architectures and beam replacements for Butler/DFT. For each final continuous-phase analog beamformer, we also evaluate the feasible digital construction in Appendix~\ref{app:rate_gap} using the fully-digital beamformer and retain the higher-rate solution. The phase-layer sweep uses $P_T=0$~dBm and $M\in\{1,\ldots,10,12,14,16,20,24,28,32,33\}$.
In the numerical implementation, refinement may also stop when the sum rate stabilizes over successive iteration windows and the digital residual is small.
The ideal-component power sweep uses $M=32$, $r=16$, and $P_T=-20,-15,\ldots,50$~dBm. Unless stated otherwise, comparisons use equal RF-chain injected power as defined in \eqref{eq:P_inj}.

\begin{figure}[!t]
\centering
\includegraphics[width=\columnwidth]{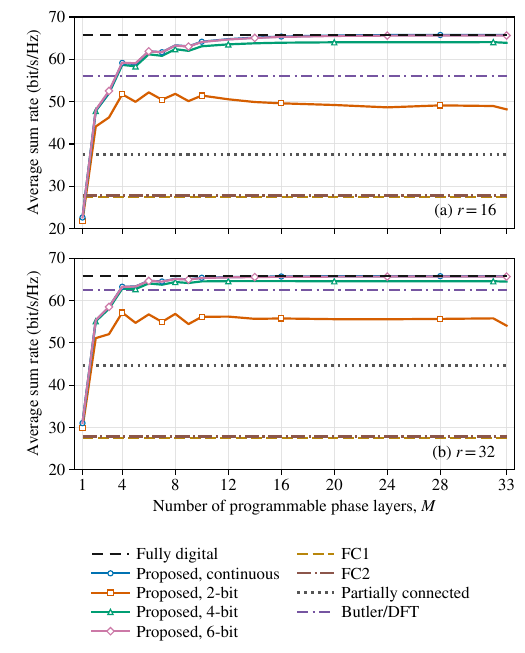}
\caption{Average sum rate versus the number of programmable phase layers at $P_{\mathrm{inj}}=0$~dBm, for $N=512$ and $S=16$: (a) $r=16$; (b) $r=32$. FC1 uses $2S=32$ RF chains in both panels; the other hybrid comparators use the panel's $r$.}
\label{fig:depth_nb}
\end{figure}

\subsection{Effect of the Number of Phase Layers and RF Chains}
\label{subsec-fig_2}
Figure~\ref{fig:depth_nb} shows the achieved sum rates for $r=16$ and $r=32$. With continuous phases, increasing the number of programmable phase layers brings the proposed architecture close to the fully-digital sum-rate optimum. Doubling the number of RF chains improves the rates at smaller phase-layer counts, allowing a comparable rate to be obtained with fewer programmable layers. The $6$-bit design closely follows the continuous-phase design, while $4$-bit control gives a modest remaining gap. With $2$-bit control, the gap is larger.
As can be seen from Figure~\ref{fig:depth_nb}, the continuous-phase rates approach the fully-digital sum-rate optimum at phase-layer counts close to the necessary bound in \eqref{eq:stream_dof_bound}, while useful rates are already obtained with substantially fewer layers. 

The comparator curves are independent of $M$. Among them, Butler/DFT beam selection achieves the highest rates, followed by the partially connected architecture. Both improve when $r$ is doubled. Their analog beamformers preserve RF-chain injected power, but their fixed beam basis or subarray connections restrict the available beamforming directions. FC1 and FC2 achieve similar rates under their common sum-rate design, with an additional penalty from their contractive analog power transfer.

\begin{figure}[!t]
\centering
\includegraphics[width=\columnwidth]{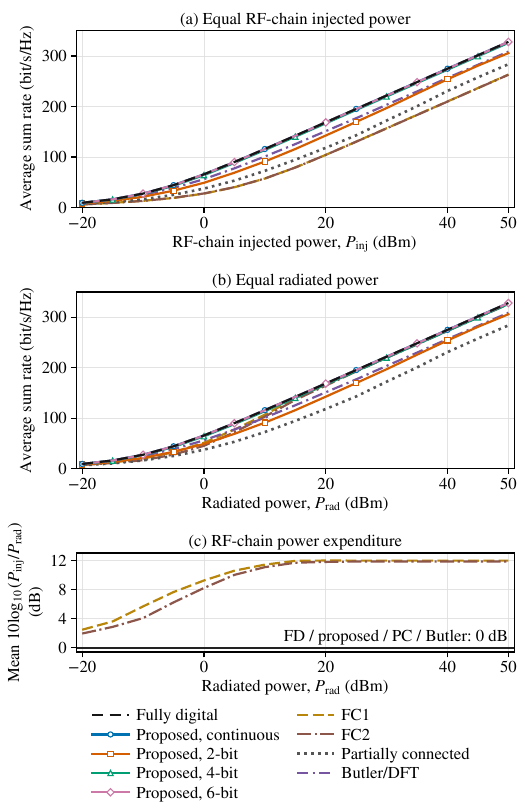}
\caption{Average sum rate under (a) equal RF-chain injected power and (b) equal radiated power, and (c) the mean per-channel $10\log_{10}(P_{\mathrm{inj}}/P_{\mathrm{rad}})$, for $N=512$, $S=r=16$, and $M=32$. Panel~(b) rescales the beamformers from panel~(a). FC1 uses 32 RF chains.}
\label{fig:pinj_nb}
\end{figure}

\subsection{Sum Rate Under Equal RF-Chain Injected Power}
\label{subsec-fig_3}
Figure~\ref{fig:pinj_nb}(a) compares the architectures at $M=32$ under equal RF-chain injected power. The proposed continuous-phase and $6$-bit designs closely follow the fully-digital sum-rate optimum throughout the power range. The $4$-bit design remains close, whereas the $2$-bit design has a more visible rate gap. All proposed variants preserve RF-chain injected power; their differences arise from the beamforming achievable with the available phase resolution.

Butler/DFT beam selection and the partially connected architecture also preserve power, but achieve lower rates because of their more restricted analog beamformers. FC1 and FC2 achieve nearly coincident rates and remain below these power-preserving comparators under equal RF-chain injected power. Thus, the ideal-component comparison exposes both the analog beamformer's ability to serve the users and the fraction of RF-chain injected power transferred to the antennas.

\subsection{Effect of the Power-Normalization Criterion}
\label{subsec:power_normalization}
Figure~\ref{fig:pinj_nb}(b) rescales the beamformers in panel~(a) to equal radiated power. This substantially narrows the apparent gap between the fully connected architectures and the proposed architecture. Panel~(c) shows the associated RF-chain power expenditure through the mean of $10\log_{10}(P_{\mathrm{inj}}/P_{\mathrm{rad}})$. The ratio is unity (0~dB) for the proposed, partially connected, and Butler/DFT architectures, while FC1 and FC2 require substantially more RF-chain injected power to provide the same radiated power. Equal-radiated-power comparison therefore conceals this additional power expenditure; the equal-RF-chain-injected-power comparison in panel~(a) exposes its effect on the achieved data rates.

\begin{figure}[!t]
\centering
\includegraphics[width=\columnwidth]{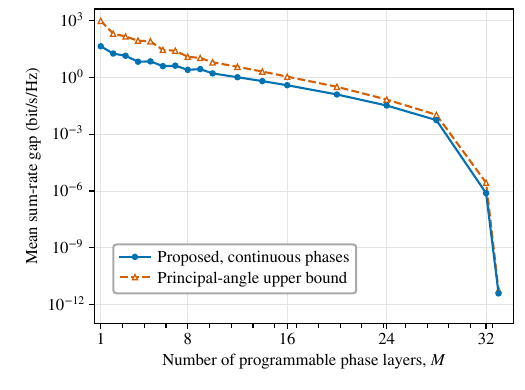}
\caption{Average continuous-phase sum-rate gap to the fully-digital sum-rate optimum and evaluated principal-angle upper bound versus the number of programmable phase layers, for $N=512$, $S=r=16$, and $P_{\mathrm{inj}}=0$~dBm.}
\label{fig:lossless_gap_bounds}
\end{figure}

\subsection{Achieved Rate Gap and Evaluated Bound}
\label{subsec:numerical_rate_gap}
Figure~\ref{fig:lossless_gap_bounds} compares the continuous-phase rate gap to the fully-digital sum-rate optimum with the principal-angle bound in \eqref{eq:principal_angle_rate_gap}. The gap and bound are evaluated on each channel before averaging.

The achieved gap decreases overall as the number of programmable phase layers increases and is numerically negligible at the largest tested count. The principal-angle bound captures the reduction associated with improved channel-subspace capture. The figure thereby connects the finite-layer analysis to the communication performance of the proposed hybrid beamforming architecture.

\subsection{Effect of Component Losses}
\label{subsec:numerical_losses}

We now compare all hybrid beamforming architectures when we take component losses into account. For all architectures, we have $N=512$, $S=r=16$, except FC1, which by design uses twice as many RF chains.

For the proposed hybrid beamforming architecture, we select $M=2$ programmable phase layers.

For the component losses, we are guided by W-band component studies and specifications in \cite{tamayo2020butler,eravant2023phase,eravant2020switch}. Thereby, we assign $0.1$~dB excess loss per elementary mixing, splitting, or combining hybrid, $1$~dB per full-range programmable phase shifter, and $0.8$~dB per binary selector element. Additional allowances are $0.1$~dB per DFT mixer for fixed phase sections and $0.01$~dB per traversed hybrid, programmable-phase, or selector stage for routing. The phase-shifter allowance represents two $180^\circ$ sections. Each DFT mixer has nine hybrid stages per route, whereas the evaluated Butler selector has thirteen binary selector stages. Corresponding components use the same loss assumptions in every architecture.

Figure~\ref{fig:component_losses} shows that the proposed architecture for $M=2$ closely follows the partially connected baseline at the low-power end and achieves a slightly higher mean data rate as power increases\footnote{At higher RF-chain injected powers, the proposed and partially connected architectures continue to achieve similar sum rates. The displayed power range in Figure~\ref{fig:component_losses} is chosen so that the rate differences are clearly visible.}. Both outperform the fully connected and Butler/DFT baselines over the displayed range. The partially connected architecture traverses fewer lossy components, whereas Butler/DFT beam selection incurs substantial selector loss because the evaluated routing implementation traverses multiple switching stages whose insertion losses accumulate. For the proposed $M=2$ architecture, attenuation accumulates through three DFT mixing layers and two programmable phase-shifting layers. The proposed architecture balances this attenuation against full-array beamforming flexibility. Comparison with the ideal-component rates at larger $M$ in Fig.~\ref{fig:depth_nb} motivates architectures with reduced mixing-layer losses, so as to gain greater beamforming flexibility by exploiting additional mixing--phase layers. 

Finally, Figure~\ref{fig:component_losses} illustrates the performance of the present architecture with component losses, while the ideal-component results in Fig.~\ref{fig:pinj_nb}(a) indicate the potential of implementations that can exploit more programmable phase layers with lower accumulated loss. An appropriate future research direction is to develop architectures better adapted to operation with lossy components and to investigate the fundamental tradeoff between accumulated loss and beamforming flexibility.

\begin{figure}[!t]
\centering
\includegraphics[width=\columnwidth]{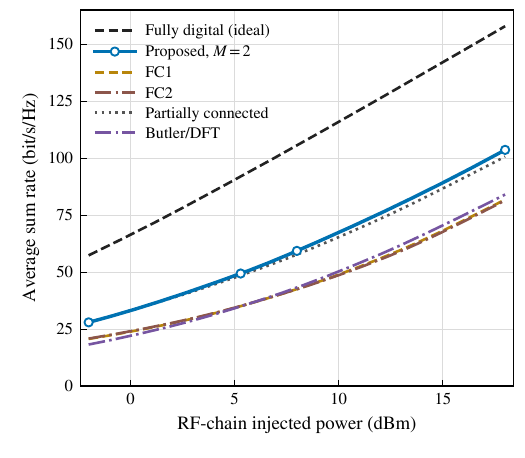}
\caption{Average sum rate with component losses for $N=512$, $S=r=16$, and the proposed architecture with $M=2$ and continuous phases. FC1 uses 32 RF chains.}
\label{fig:component_losses}
\end{figure}

\section{Conclusion}
\label{sec:conclusion}
This paper proposed a narrowband hybrid beamforming architecture with fixed mixing layers interleaved with programmable phase layers. Its semi-unitary analog beamformer transfers all RF-chain injected power to the antenna ports without loss under ideal passive-component assumptions. We designed the digital and analog beamformers and showed both theoretically and numerically that, with sufficiently many programmable phase layers, the proposed hybrid beamforming architecture attains the fully-digital sum-rate optimum under the same RF-chain injected-power constraint and ideal passive-component assumptions. In contrast, our numerical comparisons showed that none of the evaluated baseline architectures attains this optimum under the same conditions.

The numerical results show rates close to the fully-digital sum-rate optimum with sufficiently many phase layers and expose the additional RF-chain power expenditure concealed by equal-radiated-power comparisons. The component-loss example illustrates the balance between beamforming flexibility and accumulated attenuation, motivating low-loss implementations that can exploit additional programmable layers.

\appendices
\section{Proofs of the Attainability, Phase-Layer, and Rate-Gap Results}
\label{app:attainability}
\subsection{Proof of Proposition~\ref{prop:rate_containment}}
\label{app:containment}
\begin{IEEEproof}
Condition~\eqref{eq:cont_2} means that every column of $\mathbf H$ lies in the analog subspace; hence
$\mathbf F_{\mathrm{RF}}\mathbf F_{\mathrm{RF}}^H\mathbf H=\mathbf H$.
For any fully-digital optimizer $\mathbf F_{\mathrm{FD}}^\star$, choose
$\mathbf F_{\mathrm{BB}}=\mathbf F_{\mathrm{RF}}^H\mathbf F_{\mathrm{FD}}^\star$.
Then
\[
\mathbf H^H\mathbf F_{\mathrm{RF}}\mathbf F_{\mathrm{BB}}
=\mathbf H^H\mathbf F_{\mathrm{FD}}^\star,
\qquad
\|\mathbf F_{\mathrm{BB}}\|_F^2
\le\|\mathbf F_{\mathrm{FD}}^\star\|_F^2\le P_T.
\]
Thus a feasible hybrid beamformer preserves all desired and interfering received signals and attains the fully-digital sum-rate optimum. The reverse inequality follows from \eqref{eq:fd_rate_upper_bound}.
\end{IEEEproof}

\subsection{Proof of Proposition~\ref{prop:stream_dof}}
\label{app:stream_dof}
\begin{IEEEproof}
\emph{Containment required by orthogonal channels.} Take an arbitrary $\mathbf H$ with $\mathbf H^H\mathbf H=\mathbf I_S$. For every feasible fully-digital precoder,
\begin{align}
R_\Sigma(\mathbf F)
&\le\sum_s\log_2\!\left(1+\frac{|\mathbf h_s^H\mathbf f_s|^2}{\sigma^2}\right)
\nonumber\\
&\le\sum_s\log_2\!\left(1+\frac{\|\mathbf f_s\|_2^2}{\sigma^2}\right)
\le S\log_2\!\left(1+\frac{P_T}{S\sigma^2}\right).
\label{eq:orthogonal_rate_bound}
\end{align}
The last value is attained by $\mathbf f_s=\sqrt{P_T/S}\,\mathbf h_s$. Equality requires equal stream powers and equality in every Cauchy--Schwarz inequality, so every optimum has
$\mathbf f_s=\sqrt{P_T/S}\,e^{j\theta_s}\mathbf h_s$.
Hybrid attainment for all such channels therefore requires containment of every $S$-dimensional channel subspace.

\emph{Dimension count and generic noncontainment.} Let $\mathrm{Gr}(S,N)$ be the complex Grassmann manifold, of real dimension $2S(N-S)$. Fixing one phase per layer removes the common-layer phases without changing the analog subspace. The subspaces contained in realizable analog subspaces are the image of the smooth map
\[
(\Phi,\mathcal L)\longmapsto\mathbf F_{\mathrm{RF}}(\Phi)\mathcal L,
\quad
(\Phi,\mathcal L)\in\mathbb T^{M(N-1)}\times\mathrm{Gr}(S,r).
\]
Its domain has dimension $M(N-1)+2S(r-S)$. If this is smaller than $2S(N-S)$, its image has measure zero and cannot cover $\mathrm{Gr}(S,N)$. Universal attainment consequently requires
$M(N-1)\ge2S(N-r)$, which gives \eqref{eq:stream_dof_bound}.
Moreover, below this bound the map $\mathbf H\mapsto\operatorname{range}(\mathbf H)$ is a smooth submersion on the full-column-rank matrices, so the preimage of that measure-zero image also has Lebesgue measure zero (use local product coordinates and Fubini's theorem). Thus almost every full-column-rank channel fails \eqref{eq:cont_2} for every phase setting.

\emph{Strict rate loss at high SNR.} Fix any such noncontained channel. Every fully-digital sum-rate-optimal precoder uses full power and lies in $\operatorname{range}(\mathbf H)$: orthogonal projection onto this subspace preserves the received coefficients, and any power thereby saved could be reassigned by uniform scaling to increase the positive sum rate. Put
$g_s=[(\mathbf H^H\mathbf H)^{-1}]_{ss}>0$,
$h_{\max}=\max_s\|\mathbf h_s\|_2$, and $\rho=P_T/\sigma^2$.
Equal-power zero forcing achieves
$R_{\mathrm{ZF}}(\rho)=\sum_s\log_2(1+\rho/(Sg_s))$.
Since each user's rate is at most $\log_2(1+\rho h_{\max}^2)$, every fully-digital sum-rate-optimal precoder satisfies
\[
\log_2(1+\mathrm{SINR}_s)
\ge R_{\mathrm{ZF}}(\rho)-(S-1)\log_2(1+\rho h_{\max}^2)
\longrightarrow\infty.
\]
There is therefore a finite $\rho_0(\mathbf H)$ such that all these SINRs exceed $S-1$ for every $\rho>\rho_0(\mathbf H)$, uniformly over all fully-digital optimizers. For the received coefficient matrix $\mathbf Y=\mathbf H^H\mathbf F_{\mathrm{FD}}^\star$, this implies
\[
|Y_{ss}|^2>(S-1)\sum_{j\ne s}|Y_{sj}|^2
\ge\left(\sum_{j\ne s}|Y_{sj}|\right)^2.
\]
Hence $\mathbf Y$ is strictly diagonally dominant and nonsingular, also when $S=1$. Every fully-digital sum-rate-optimal precoder then has rank $S$ and column space $\operatorname{range}(\mathbf H)$.
The hybrid feasible set is compact, so its maximum is attained. Equality with the fully-digital sum-rate optimum at any $\rho>\rho_0(\mathbf H)$ would make its composite precoder a fully-digital optimizer and imply \eqref{eq:cont_2}, contradicting noncontainment.
\end{IEEEproof}

\subsection{Proof of Proposition~\ref{prop:local_sufficiency}}
\label{app:local_sufficiency}
\begin{IEEEproof}
Use indices in $\mathbb Z_N=\{0,\ldots,N-1\}$, let $\omega=e^{j2\pi/N}$, and write
$[\mathbf W]_{mn}=N^{-1/2}\omega^{-mn}$, $\mathbf w_m=\mathbf W\mathbf e_m$, and
$\mathbf\Pi=\mathbf W^2$, the permutation $n\mapsto-n$. All port indices below are modulo $N$.
We first exhibit a regular point at $M=2S$ of the map
$(\Phi,\mathcal L)\mapsto\mathbf F_{\mathrm{RF}}(\Phi)\mathcal L$, where
$\mathcal L\in\mathrm{Gr}(S,r)$.

\emph{Tangent-space criterion.} Let $\mathbf E_S$ select the first $S$ ports and $\mathbf E_\perp$ the last $N-r$ ports. At
$\mathcal L=\operatorname{span}(\mathbf e_0,\ldots,\mathbf e_{S-1})$, put
$\mathbf V_k=\mathbf D_k\mathbf W\cdots\mathbf D_1\mathbf W$. After applying $\mathbf X^H$ to the tangent space,
variations of $\mathcal L$ supply all directions between the stream subspace and the other driven ports. The remaining
phase directions are the matrices
\[
j(\mathbf E_\perp^H\mathbf V_k^H\mathbf e_n)
(\mathbf e_n^H\mathbf V_k\mathbf E_S)
\in\mathbb C^{(N-r)\times S}.
\]
This follows by differentiating $\mathbf X$:
$\partial\mathbf X/\partial\phi_{n,k}=j\mathbf X\mathbf V_k^H\mathbf e_n\mathbf e_n^T\mathbf V_k$.
Their real span is the whole space if and only if the only matrix
\[
\mathbf\Delta=\mathbf E_S\mathbf Y^H\mathbf E_\perp^H+
\mathbf E_\perp\mathbf Y\mathbf E_S^H,
\qquad \mathbf Y\in\mathbb C^{(N-r)\times S},
\]
satisfying
\begin{equation}
\operatorname{diag}(\mathbf V_k\mathbf\Delta\mathbf V_k^H)=\mathbf0,
\qquad k=1,\ldots,2S,
\label{eq:local_dual_condition}
\end{equation}
is $\mathbf\Delta=\mathbf0$. Indeed, each diagonal entry is twice the real inner product of
$j\mathbf Y$ with the corresponding phase direction. Removing a common phase from each layer enforces the phase convention of Section~\ref{sec:beamformer_design} without changing this map or its rank.

\emph{Phase-mask construction.} Set $\epsilon=j-1$ and choose
\begin{equation}
\begin{split}
\mathbf D_k&=\mathbf I_N+\epsilon\mathbf e_{m_k}\mathbf e_{m_k}^T,\\
m_{2i-1}&=(-1)^{i-1}(i-1),\qquad m_{2i}=(-1)^i u_0,
\end{split}
\label{eq:local_sparse_masks}
\end{equation}
for $i=1,\ldots,S$, where $u_0=S$ if $S\ge2$, and $u_0=0$ if $S=1$.
The assumption $r\ge S+1$ when $S\ge2$ makes $u_0$ a driven port carrying no stream.
Every mask has only one nonzero phase, $\pi/2$.

Suppose \eqref{eq:local_dual_condition} holds, and define $\mathbf\Delta_0=\mathbf\Delta$ and
$\mathbf\Delta_k=\mathbf D_k\mathbf W\mathbf\Delta_{k-1}\mathbf W^H\mathbf D_k^H$.
For a Hermitian matrix $\mathbf A$, write
\[
c^{(m)}(\mathbf A)_n=\sum_p\omega^{-m(p-n)}A_{pn}
=\sqrt N\omega^{mn}(\mathbf w_m^T\mathbf A)_n.
\]
The identity $\mathbf W\mathbf D_k\mathbf W=\mathbf\Pi+\epsilon\mathbf w_{m_k}\mathbf w_{m_k}^T$ gives
\begin{equation}
\operatorname{Re}\{\epsilon c^{(m_k)}(\mathbf\Delta_{k-1})_n\}=0,
\quad k=1,\ldots,2S-1,
\label{eq:local_column_condition}
\end{equation}
for all $n$. To see this, expand the diagonal of
$\mathbf W\mathbf\Delta_k\mathbf W^H$. Its term independent of $\epsilon$ is a permuted diagonal of
$\mathbf\Delta_{k-1}$, and its $|\epsilon|^2$ term is proportional to
$[\mathbf W\mathbf\Delta_{k-1}\mathbf W^H]_{m_km_k}$; both vanish by
\eqref{eq:local_dual_condition} and the zero diagonal of $\mathbf\Delta_0$.
The cross term is $2N^{-1}\operatorname{Re}\{\epsilon c^{(m_k)}(\mathbf\Delta_{k-1})_{-n}\}$, and the
whole diagonal vanishes by the condition at $k+1$.

\emph{Inductive elimination.}
We prove inductively that
\begin{equation}
\mathbf\Delta_{2i-2}=\mathbf\Pi^{i-1}\mathbf\Delta\mathbf\Pi^{i-1},
\qquad \mathbf w_l^T\mathbf\Delta=\mathbf0\quad(0\le l<i-1).
\label{eq:local_induction}
\end{equation}
The assertion is immediate at $i=1$. For $i<S$, put $\mathbf A=\mathbf\Delta_{2i-2}$,
$\mu=m_{2i-1}$, $\nu=m_{2i}$, and $\widetilde{\mathbf A}=\mathbf W\mathbf A\mathbf W^H$.
The row $-\nu$ of $\mathbf A$ vanishes, since it corresponds to row $u_0$ of $\mathbf\Delta$.
By \eqref{eq:local_column_condition}, $c^{(\mu)}(\mathbf A)_n=j\bar\epsilon t_n$ for a real vector $\mathbf t$.
Its inverse DFT $\tau(l)=N^{-1/2}\sum_n t_n\omega^{ln}$ satisfies $\tau(-l)=\overline{\tau(l)}$ and
$\widetilde A_{\mu n}=j\bar\epsilon N^{-1/2}\tau(n-\mu)$.
The diagonal entry at $n=\mu$ is real, whereas $j\bar\epsilon=1-j$, so $\tau(0)=0$.
The zero row $-\nu$ gives $c^{(\nu)}(\widetilde{\mathbf A})=0$, using
$\mathbf w_\nu^T\mathbf W=\mathbf e_{-\nu}^T$. Consequently, for $n\ne\mu$,
$c^{(\nu)}(\mathbf\Delta_{2i-1})_n
=\epsilon\omega^{\nu(n-\mu)}\widetilde A_{\mu n}$.
Applying \eqref{eq:local_column_condition} at $k=2i$ shows that
$\epsilon\omega^{\nu l}\tau(l)$ is real for $l\ne0$; applying it at $-l$ and conjugating shows that
$\bar\epsilon\omega^{\nu l}\tau(l)$ is real. Since $\epsilon/\bar\epsilon=-j$, both statements force
$\tau(l)=0$. Thus $\mathbf w_\mu^T\mathbf A=0$, giving $\mathbf w_{i-1}^T\mathbf\Delta=0$.
Row and column $\mu$ of $\widetilde{\mathbf A}$ vanish, so the first mask has no effect. The second also has no effect
on $\mathbf\Pi^i\mathbf\Delta\mathbf\Pi^i$, whose row and column $\nu$ correspond to $u_0$.
This proves \eqref{eq:local_induction} at $i+1$.

\emph{Elimination of the remaining term.}
Write $\mathbf\Delta=\sum_{q<S}(\mathbf e_q\mathbf v_q^H+\mathbf v_q\mathbf e_q^H)$, where each $\mathbf v_q$ is
supported on $\mathcal T=\{r,\ldots,N-1\}$. At a terminated column $n$, the relations
$\mathbf w_l^T\mathbf\Delta=0$ for $l=0,\ldots,S-2$ state that
$\sum_{q<S}\overline{v_q(n)}z^q$ vanishes at $z=\omega^{-l}$. Therefore, with
$\xi(z)=\prod_{l=0}^{S-2}(z-\omega^{-l})=\sum_{q<S}\xi_qz^q$ and
$\boldsymbol\xi=\sum_{q<S}\xi_q\mathbf e_q$, we have
\[
\mathbf\Delta=\boldsymbol\xi\mathbf x^H+\mathbf x\boldsymbol\xi^H,
\qquad (\mathbf W\mathbf x)(l)=0\quad(0\le l<S-1),
\]
for a vector $\mathbf x$ supported on $\mathcal T$. The last equalities follow by evaluating
$\mathbf w_l^T\mathbf\Delta=0$ at column $S-1$, since $\xi_{S-1}=1$.
For $S=1$ this representation holds directly with $\xi(z)=1$ and $\boldsymbol\xi=\mathbf e_0$.

Put $\mu=S-1$ and $c=\omega^{-\mu}$. By \eqref{eq:local_induction} and index reversal,
\eqref{eq:local_column_condition} at $k=2S-1$ becomes
$\operatorname{Re}\{\epsilon c^{(\mu)}(\mathbf\Delta)_n\}=0$.
For $n\in\mathcal T$, this is $\operatorname{Re}\{\epsilon\omega^{\mu n}\overline{x_n}\xi(c)\}=0$.
Since $\xi(c)\ne0$, write $x_n=\lambda\omega^{\mu n}s_n$, where
$\lambda=-j/\overline{\epsilon\xi(c)}$ and $\mathbf s$ is real and supported on $\mathcal T$.
Then $\mathbf W\mathbf x=\lambda\widehat{\mathbf s}(\cdot-\mu)$, where
$\widehat{\mathbf s}=\mathbf W\mathbf s$ satisfies
$\widehat s(-l)=\overline{\widehat s(l)}$ and $\widehat s(l-\mu)=0$ for $0\le l<S-1$.
Condition \eqref{eq:local_dual_condition} at $k=1$ further gives
\[
\operatorname{Re}\{\psi(l)\overline{\widehat s(l)}\}=0,
\qquad \psi(l)=(\mathbf W\boldsymbol\xi)(\mu+l)\bar\lambda.
\]
If $\psi(l)$ or $\psi(-l)$ vanishes, the preceding Fourier zeros and conjugate symmetry already give
$\widehat s(l)=0$. Otherwise the equations at $l$ and $-l$ imply
$\widehat s(l)=j\psi(l)a$ and $\overline{\widehat s(l)}=j\psi(-l)b$ for real $a,b$.
A nonzero $\widehat s(l)$ would therefore require $\psi(l)\psi(-l)$ to be real. However,
\[
\psi(l)\psi(-l)=-\frac{1}{N\epsilon^2}
\frac{\xi(ce^{-j2\pi l/N})\xi(ce^{j2\pi l/N})}{\xi(c)^2}
\]
is nonzero and purely imaginary: $\epsilon^2=-2j$, and the fraction is real because each unit-modulus root $z$ of
$\xi$ contributes the real factor
$(c/z+z/c-2\cos(2\pi l/N))/(c/z+z/c-2)$.
Thus $\widehat{\mathbf s}=0$, then $\mathbf x=0$ and $\mathbf\Delta=0$.

\emph{Open-set attainment.} The constructed point is regular, so the submersion theorem gives a nonempty open set of attainable
$S$-dimensional channel subspaces. Appending any fixed phase layer multiplies $\mathbf X$ on the left by
$\mathbf W\mathbf D$, a unitary transformation; it carries this open set to an open attainable set with one additional phase layer.
This proves the assertion for every $M\ge2S$. The inverse image under
$\mathbf H\mapsto\operatorname{range}(\mathbf H)$ is a nonempty open set of full-column-rank channels, and
Proposition~\ref{prop:rate_containment} gives fully-digital sum-rate attainment for every positive power and noise level.
\end{IEEEproof}

\subsection{Proof of Proposition~\ref{prop:rate_gap}}
\label{app:rate_gap}
\begin{IEEEproof}
Fix a feasible phase setting with $\psi_{\max}<\pi/2$. Let $\mathbf T$ be an orthonormal basis of $\operatorname{range}(\mathbf H)$, and choose a fully-digital optimizer in this subspace, written as $\mathbf F_{\mathrm{FD}}^\star=\mathbf T\mathbf G$ with $\|\mathbf G\|_F^2=P_T$.
Such an optimizer exists because projection onto the channel subspace preserves every received coefficient, and uniform scaling to full power cannot reduce the rate.
Let $\mathbf C_{\mathrm{RF}}=\mathbf T^H\mathbf F_{\mathrm{RF}}(\Phi)$, which has full row rank and squared singular values $\cos^2\psi_i\le1$.
Define the feasible digital beamformer
\[
\mathbf F_{\mathrm{BB}}=\sqrt\kappa\,\mathbf C_{\mathrm{RF}}^\dagger\mathbf G,
\qquad
\kappa=\frac{P_T}{\|\mathbf C_{\mathrm{RF}}^\dagger\mathbf G\|_F^2}
\in[\cos^2\psi_{\max},1].
\]
It has squared norm $P_T$ and satisfies
$\mathbf H^H\mathbf F_{\mathrm{RF}}\mathbf F_{\mathrm{BB}}
=\sqrt\kappa\,\mathbf H^H\mathbf F_{\mathrm{FD}}^\star$.
Thus desired and interfering signal powers are both scaled by $\kappa$. If $\gamma_s$ is the fully-digital SINR, the constructed hybrid SINR is at least $\kappa\gamma_s$, since $\kappa\le1$ and the noise power is unchanged. Its rate cannot exceed the globally optimized hybrid rate, giving
\begin{align*}
0&\le R_{\mathrm{FD}}^\star(\mathbf H)
-R_{\mathrm{HYB}}^\star(\mathbf H;M)\\
&\le\sum_s\log_2\frac{1+\gamma_s}{1+\kappa\gamma_s}\\
&\le S\log_2\frac1\kappa
\le S\log_2\sec^2\psi_{\max}.
\end{align*}
Finally, $\sin^2\psi_{\max}\le\sum_i\sin^2\psi_i=\delta(\Phi)$.
When the unnormalized deficit satisfies $\delta(\Phi)<1$, this gives
$\cos^2\psi_{\max}\ge1-\delta(\Phi)>0$ and the additional bound
$S\log_2[1/(1-\delta(\Phi))]$.
\end{IEEEproof}

The same construction applies to any feasible fully-digital beamformer with positive rate, using its channel-subspace projection and the projected power in the definition of $\kappa$. Scaling the resulting hybrid candidate to full power cannot reduce its rate. If this candidate is retained whenever it improves the implemented hybrid beamformer, the latter's rate shortfall from that fully-digital beamformer is also bounded by $S\log_2\sec^2\psi_{\max}$.

\section{Scalar Phase Optimization with Digital Readaptation}
\label{app:quantized_reduced}
We give a proof sketch of Proposition~\ref{prop:quantized_phase_update} and the expressions needed to evaluate its update. Hold $u_s^{(t)},w_s^{(t)}$ fixed at the values for the positive-rate pair starting the sweep, and suppress $t$ below.

\subsection{Scalar Objective and Its Minimum}
Changing one phase, with $z=e^{j\phi}$ and current value $z_0$, gives the rank-one effective-channel change
\begin{equation}
\mathbf H_{\mathrm{eff}}(z)=\mathbf H_{\mathrm{eff}}(z_0)
+(z-z_0)\boldsymbol\ell\mathbf b^H,
\label{eq:quantized_rank_one}
\end{equation}
where $\boldsymbol\ell$ and $\mathbf b$ collect the fixed downstream and upstream transfers for that phase. Define
\begin{align}
\mathbf K(z)&=\mathbf H_{\mathrm{eff}}(z)^H
\operatorname{diag}(w_s|u_s|^2)\mathbf H_{\mathrm{eff}}(z)+\nu\mathbf I_r,\nonumber\\
\mathbf V(z)&=\mathbf H_{\mathrm{eff}}(z)^H\operatorname{diag}(w_su_s),
\quad \nu=\frac{\sigma^2}{P_T}\sum_s w_s|u_s|^2>0.
\label{eq:quantized_normal_matrices}
\end{align}
Eliminating the digital beamformer by \eqref{eq:wmmse_digital_update} gives
\begin{equation}
\widehat{\mathcal J}(z)=\sum_s(w_s-\ln w_s)
-\operatorname{tr}\!\left(\mathbf V(z)^H\mathbf K(z)^{-1}\mathbf V(z)\right).
\label{eq:quantized_reduced_value}
\end{equation}
On the unit circle, $\mathbf K(z)$ is a constant matrix plus two rank-one terms proportional to $z$ and $z^*$. Expanding its determinant and the bordered determinants for the quadratic forms in \eqref{eq:quantized_reduced_value} leaves only $1,z,z^*$: higher powers vanish by rank-one dependence, and $zz^*=1$. Thus \eqref{eq:quantized_reduced_value} has the first-harmonic ratio form \eqref{eq:quantized_reduced_scalar}. Since $\mathbf K(z)\succeq\nu\mathbf I_r$, its determinant is positive for every phase, giving the stated denominator condition.

Use three auxiliary angles $\theta_\ell=2\pi\ell/3$, $\ell=0,1,2$, and set
\begin{equation}
d_\ell=\frac{\det\mathbf K(e^{j\theta_\ell})}{\det\mathbf K(z_0)},
\qquad p_\ell=d_\ell\widehat{\mathcal J}(\theta_\ell).
\label{eq:quantized_coefficient_samples}
\end{equation}
For either sample set $x_\ell$, the coefficients are
\begin{equation}
\begin{split}
x_{\mathrm{const}}&=(x_0+x_1+x_2)/3,\quad
x_{\mathrm c}=(2x_0-x_1-x_2)/3,\\
x_{\mathrm s}&=(x_1-x_2)/\sqrt3.
\end{split}
\label{eq:quantized_coefficient_recovery}
\end{equation}
Apply \eqref{eq:quantized_coefficient_recovery} to $p_\ell$ for the numerator and $d_\ell$ for the denominator. The derivative of \eqref{eq:quantized_reduced_scalar} has the sign of $\zeta_{\mathrm c}\cos\phi+\zeta_{\mathrm s}\sin\phi+\zeta_0$, where
\begin{align}
\zeta_{\mathrm c}&=a_{\mathrm s}b_0-a_0b_{\mathrm s},\qquad
\zeta_{\mathrm s}=a_0b_{\mathrm c}-a_{\mathrm c}b_0,\nonumber\\
\zeta_0&=a_{\mathrm s}b_{\mathrm c}-a_{\mathrm c}b_{\mathrm s}.
\label{eq:quantized_stationary_coefficients}
\end{align}
For a nonconstant objective, $\zeta=(\zeta_{\mathrm c}^2+\zeta_{\mathrm s}^2)^{1/2}>0$ and
$\zeta_0=(b_{\mathrm c}\zeta_{\mathrm c}+b_{\mathrm s}\zeta_{\mathrm s})/b_0$ implies $|\zeta_0|<\zeta$. The continuous-phase minimizer is therefore
\begin{equation}
\phi_\star=\operatorname{atan2}(\zeta_{\mathrm s},\zeta_{\mathrm c})
-\arccos(-\zeta_0/\zeta)\quad(\mathrm{mod}\ 2\pi).
\label{eq:quantized_continuous_minimum}
\end{equation}
The derivative changes from negative to positive there, and its other zero gives the maximum. Monotonicity along the two arcs between these extrema places a grid minimum at one of the two neighbors of $\phi_\star$, proving \eqref{eq:quantized_phase_update}. A constant objective leaves the current phase unchanged.

Each coordinate decreases or preserves the reduced objective. The acceptance argument in Appendix~\ref{app:joint_improvement_proof} then ensures that the completed, normalized sweep preserves RF-chain injected power and does not decrease the starting rate.

\subsection{Evaluation and Candidate Handling}
\label{app:quantized_cost}
Evaluate \eqref{eq:quantized_reduced_value} through the residual of its regularized least-squares problem to avoid cancellation. Reusing rank-one QR updates for \eqref{eq:quantized_rank_one} costs $\mathcal O((r+S)^2)=\mathcal O(r^2+rS)$ per trial; determinant ratios follow from the triangular factors. Together with forward and downstream DFT propagation, this gives \eqref{eq:quantized_sweep_complexity}. Rebuild the factors at each layer to limit roundoff accumulation.

\label{app:quantized_candidates}
For duplicate phase settings, retain the inherited digital beamformer with the highest rate. Refine positive-rate starts; retain a zero-rate inherited start without refinement, and skip refinement for a zero effective channel. Reset $t=0$ for each digital start, then continue the selected pair's counter: each accepted digital update or completed normalized sweep increments $t$, while coordinate changes leave it fixed. Retain the current phase on numerical ties.

\section{Derivation Sketch and Implementation Details for Beamformer Design}
\label{app:continuous_design}
\subsection{Digital Update}
\label{app:digital_update_proof}
For a fixed physical beamformer, let $e_s=\mathbb E|s_s-u_s^*y_s|^2$. The WMMSE relation is
\begin{equation}
\min_{\{u_s\in\mathbb C,\,w_s>0\}}
\sum_s(w_se_s-\ln w_s)
=S-(\ln2)R_\Sigma.
\label{eq:wmmse_rate_equivalence}
\end{equation}
Its minimizing receivers and weights give \eqref{eq:wmmse_receiver}--\eqref{eq:wmmse_weight}. For any nonzero auxiliary matrix, substituting $\mathbf F_{\mathrm{BB}}=\sqrt{P_T}\widetilde{\mathbf F}_{\mathrm{BB}}/\|\widetilde{\mathbf F}_{\mathrm{BB}}\|_F$ into \eqref{eq:sim_sinr} gives
\begin{equation}
\label{eq:wmmse_scaled_sinr}
\mathrm{SINR}_s=
\frac{|[\mathbf H_{\mathrm{eff}}(\Phi)\widetilde{\mathbf F}_{\mathrm{BB}}]_{ss}|^2}
{\displaystyle\sum_{j\ne s}|[\mathbf H_{\mathrm{eff}}(\Phi)\widetilde{\mathbf F}_{\mathrm{BB}}]_{sj}|^2
+\frac{\sigma^2}{P_T}\|\widetilde{\mathbf F}_{\mathrm{BB}}\|_F^2}.
\end{equation}
For $\beta=\sqrt{P_T}/\|\widetilde{\mathbf F}_{\mathrm{BB}}\|_F$, the virtual signal $y_s/\beta$ has the SINR above and the auxiliary MSE in \eqref{eq:wmmse_virtual_mse}. Its optimal scalar receive coefficient is $\beta$ times the physical one, so minimizing over the receivers and weights gives the same WMMSE--rate identity. Choosing the current full-power digital beamformer as the auxiliary matrix gives $\beta=1$. Holding $u_s^{(t)},w_s^{(t)}$ fixed makes the weighted MSE quadratic in $\widetilde{\mathbf F}_{\mathrm{BB}}$. Setting its derivative to zero gives \eqref{eq:wmmse_digital_update}. The coefficient matrix is positive definite for a positive-rate current pair, so this minimizer is unique.

\subsection{Phase Gradient and Acceptance}
\label{app:joint_improvement_proof}
At the digital minimizer, the derivative with respect to the digital matrix is zero. Thus the derivative of $\widehat{\mathcal J}^{(t)}(\Phi)$ is obtained by holding that minimizer fixed. For $L=-\widehat{\mathcal J}^{(t)}$, the convention
$\mathrm dL=\Re\{\operatorname{tr}((\mathbf Y_{M+1}^{(t)})^H\mathrm d\mathbf Z_M^{(t)})\}$ gives \eqref{eq:reduced_wmse_adjoint_seed}. Propagating this differential backward through \eqref{eq:fwd_rec} gives \eqref{eq:bwd_rec}; using $\mathrm d e^{j\phi}=j e^{j\phi}\mathrm d\phi$ then gives \eqref{eq:vec_grad}.

Digital minimization, the accepted phase step, and reoptimizing the receivers and weights imply
\begin{align*}
S-(\ln2)R_\Sigma^{(t+1)}
&\le\widehat{\mathcal J}^{(t)}(\Phi^{(t+1)})\\
&\le\widehat{\mathcal J}^{(t)}(\Phi^{(t)})
\le S-(\ln2)R_\Sigma^{(t)}.
\end{align*}
Hence the accepted update does not decrease the rate. Its auxiliary matrix is nonzero: at zero the weighted MSE is $\sum_s(w_s^{(t)}-\ln w_s^{(t)})\ge S$, whereas the chain is strictly below $S$ for a positive-rate current pair. The same argument applies with the phases held fixed.

\subsection{Candidate Generation and DFT Continuation}
\label{app:capture_details}
For the positive capture score $\|\mathbf T^H\mathbf F_{\mathrm{RF}}\|_F^2$, use $\mathbf Z_0=\mathbf W\mathbf E_r$ in \eqref{eq:fwd_rec} and the adjoint seed
\begin{equation}
\label{eq:capture_adjoint_seed}
\mathbf Y_{M+1}=2\mathbf T\mathbf T^H\mathbf Z_M.
\end{equation}
Equations~\eqref{eq:bwd_rec} and \eqref{eq:vec_grad}, with iteration indices suppressed and $\mathbf1_S$ replaced by $\mathbf1_r$, give its gradient. L-BFGS uses the negative of this gradient to minimize the negative score in \eqref{eq:capture_initialization}.

For DFT continuation, start with $m_0=1$ for odd $M$ or $m_0=2$ for even $M$ and use
\begin{equation}
\mathbf T_m=\mathbf W^{m-M}\mathbf T,
\qquad m=m_0,m_0+2,\ldots,M.
\label{eq:capture_continuation_target}
\end{equation}
At each count $m$, optimize the capture score with target $\mathbf T_m$. Appending two zero-phase layers multiplies the analog transfer by $\mathbf W^2$; changing the target to $\mathbf T_{m+2}=\mathbf W^2\mathbf T_m$ preserves the score. Use the padded setting as the next seed, reoptimizing all phases and retaining it if unimproved. The final stage is \eqref{eq:capture_initialization} with target $\mathbf T$ and its entering setting as $\Phi_{\mathrm{seed},i}$. Capture candidates can be reused across power settings.

\subsection{Digital Convergence Safeguard}
\label{app:digital_safeguard}
If the fixed-phase digital budget is exhausted with a significant residual \eqref{eq:digital_update_residual}, apply L-BFGS to the rate over the real and imaginary entries of an unnormalized digital matrix, starting from $\mathbf F_{\mathrm{BB}}^{(t)}$. Normalize each nonzero trial matrix to Frobenius norm $\sqrt{P_T}$. Accept a rate-improving output as a digital-only update under the existing iteration convention; otherwise keep the current pair.

For this safeguard, the physical sum-rate gradient is
\begin{align}
\nabla_{\mathbf F}R_\Sigma
=\frac{2}{\ln2}\mathbf H\big[\operatorname{diag}(w_su_s)
-\operatorname{diag}(w_s|u_s|^2)\mathbf H^H\mathbf F\big],
\label{eq:rate_adjoint_seed}
\end{align}
with $u_s,w_s$ evaluated at the normalized composite precoder. Multiplication by $\mathbf F_{\mathrm{RF}}^H$ and the chain rule for normalization give the gradient with respect to the unnormalized digital matrix; its real and imaginary parts are used by L-BFGS.

\bibliographystyle{IEEEtran}
\bibliography{refs}

\end{document}